\documentclass[conference]{IEEEtran}
\usepackage{epsfig,epstopdf} 
\usepackage{amsmath,amssymb,amsthm,cite,graphicx,array}
\usepackage{algorithm}
\usepackage[noend]{algpseudocode}
\theoremstyle{definition}
\newtheorem{theorem}{Theorem}
\newtheorem{lemma}{Lemma}
\newtheorem{corollary}{Corollary}

\newtheorem{definition}{Definition}

\newtheorem{example}{Example}
\usepackage{color}

\title{Low-Complexity Decoding for Symmetric, Neighboring and Consecutive Side-information Index Coding Problems}

\begin{document}

\author{Mahesh~Babu~Vaddi~and~B.~Sundar~Rajan\\ 
 Department of Electrical Communication Engineering, Indian Institute of Science, Bengaluru 560012, KA, India \\ E-mail:~\{mahesh,~bsrajan\}@ece.iisc.ernet.in }
 
\maketitle
\begin{abstract}
The capacity of symmetric, neighboring and consecutive side-information single unicast index coding problems (SNC-SUICP) with number of messages equal to the number of receivers was given by Maleki, Cadambe and Jafar. For these index coding problems, an optimal index code construction by using Vandermonde matrices was proposed. This construction requires all the side-information at the receivers to decode their wanted messages and also requires large field size. In an  earlier work, we constructed binary matrices of size $m \times n (m\geq n)$ such that any $n$ adjacent rows of the matrix are linearly independent over every field. Calling these matrices as Adjacent Independent Row (AIR) matrices using which we gave an optimal scalar linear index code for the one-sided SNC-SUICP for any given number of messages and one-sided side-information. By using Vandermonde matrices or AIR matrices, every receiver needs to solve $K-D$ equations with $K-D$ unknowns to obtain its wanted message, where $K$ is the number of messages and $D$ is the size of the side-information. In this paper, we analyze some of the combinatorial properties of the AIR matrices. By using these properties, we present a low-complexity decoding which helps to identify a reduced set of side-information for each users with which the decoding can be carried out. By this method every receiver is able to decode its wanted message symbol by simply adding some index code symbols (broadcast symbols). We explicitly give both the reduced set of side-information and the broadcast messages to be used  by each receiver to decode its wanted message. For a given  pair or receivers our decoding identifies  which one will perform better than the other when the broadcast channel is noisy.  


\end{abstract}
\section{Introduction and Background}
\label{sec1}
In an index coding problem, there is a unique source $S$ having a set of $K$ messages $X= \lbrace x_{0},x_{1},\ldots,x_{K-1} \rbrace $ and a set of $m$ receivers $\mathcal{R}=\lbrace R_{0},R_{1},\ldots,R_{m-1} \rbrace$. The messages $x_{k}$, $k \in \{0,1,2,\cdots,K-1\}=[0:K-1]$ take values from some finite field $\mathbb{F}_q$. Each receiver $R_{k} \in \mathcal{R}$ is specified by the tuple $(\mathcal{W}_{k},\mathcal{K}_k)$, where $\mathcal{W}_{k} \subseteq X$ are the messages demanded by $R_{k}$ and $\mathcal{K}_{k} \subseteq X \setminus \mathcal{W}_{k}$ is the information known at the receiver called as side-information of the receiver. An index coding problem is completely specified by $(X,\mathcal{R})$ and it is referred as $\mathcal{I}(X,\mathcal{R})$. 

An \textit{index code} over $\mathbb{F}_{q}$ for an instance of the index coding problem $\mathcal{I}(X,\mathcal{R})$, is an encoding function $\mathfrak{E}:\mathbb{F}_{q}^{K} \rightarrow \mathbb{F}_{q}^N$ such that for each receiver $R_{k}$, $k\in [0:m-1]$, there exists a decoding function $\mathfrak{D}_{k}:\mathbb{F}_{q}^{N} \times \mathbb{F}_{q}^{|\mathcal{K}_{k}|}\rightarrow \mathbb{F}_{q}^{|{\mathcal{W}_{k}}|} $ satisfying $ \mathfrak{D}_{k}(\mathfrak{E}(X),\mathcal{K}_{k})=\mathcal{W}_{k} , \forall \; X \in \mathbb{F}_{q}^{K}$. An index code is said to be \textit{linear} if the encoding function $\mathfrak{E}$ is  linear. The parameter $N$ is called the \textit{length} of the index code. 

Birk and Kol \cite{BiK} introduced the problem of index coding with side-information. Ong and Ho \cite{OnH} classify the binary index coding problem depending on the demands and the side-information possessed by the receivers. An index coding problem is unicast if the demand sets of the receivers are disjoint. An index coding problem is single unicast if the demand sets of the receivers are disjoint and the cardinality of demand set of every receiver is one. Any unicast index problem can be converted into a single unicast index coding problem. Bar-Yossef \textit{et al.} \cite{YBJK} studied single unicast index coding problems. It was found that the length of an optimal linear index code is equal to the minrank of the side-information graph of the index coding problem but finding the minrank is NP hard \cite{minrank}.

Maleki \textit{et al.} \cite{MCJ} found the capacity of symmetric neighboring and consecutive side-information single unicast index coding problem (SNC-SUICP). In a SNC-SUICP with equal number of $K$ messages and source-destination pairs, each destination has a total of $U+D=A<K$ side-information, corresponding to the $U$ messages before and $D$ messages after its desired message. In this setting, the $k-$th receiver $R_{k}$ demands the message $x_{k}$ having the side-information
\begin{equation}
\label{side-information}
\{x_{k-U},\dots,x_{k-2},x_{k-1}\}~\cup~\{x_{k+1}, x_{k+2},\dots,x_{k+D}\}.
\end{equation}
The symmetric capacity of this index coding problem is:
\begin{equation}
\label{capacity}
C=\left\{
                \begin{array}{ll}
                  {1,\qquad\quad\ A=K-1}\\
                  {\frac{U+1}{K-A+2U}},A\leq K-2\qquad \text{per message},
                  \end{array}
              \right.
\end{equation}
               where $U,D \in$ $\mathbb{Z},$ $0 \leq U \leq D$, and $U+D=A<K$. In the setting of \cite{MCJ} with one sided side-information cases, i.e., the cases where $U$ is zero,
\begin{equation}
\label{side-informationmcj}
{\cal K}_k =\{x_{k+1}, x_{k+2},\dots,x_{k+D}\}, 
\end{equation}
\noindent
for which \eqref{capacity} reduces to
\begin{equation}
\label{capacity1}
C=\left\{
                \begin{array}{ll}
                  {1 ~~~~~~~~~~~~ \mbox{if} ~~ D=K-1}\\
                  {\frac{1}{K-D}} ~~~~~~ \mbox{if} ~~D\leq K-2, 
                  \end{array}
              \right.
\end{equation}
symbols per message.

The SNC-SUICP considered in this paper was referred as symmetric neighboring antidote multiple unicast index coding problem by Maleki, Cadambe and Jafar in \cite{MCJ}.
\subsection{Motivation and Contributions}
Maleki \textit{et al.} \cite{MCJ} proposed index code construction by using Vandermonde matrices. However the index code construction procedure by using Vandermonde matrices require large field size and the field size in this method depends on the number of messages $K$. This method also requires every receiver to use all of its $D$ side-information to decode its wanted message. The problem of finding an optimum length index code is known to be NP-hard \cite{YBJK,ChS} and the complexity of finding an optimal solution increases exponentially with number of messages and the number of side-information. In \cite{VaR2}, we gave a construction of binary matrices with a given size $m \times n~(m \geq n)$, such that any $n$ adjacent rows in the matrix are linearly independent over every field $\mathbb{F}_q$. We refer these matrices as Adjacent Independent Row (AIR) matrices . By using AIR matrices, we presented capacity achieving scalar linear codes for given $K$ and $D$ over every field $\mathbb{F}_q$, which is independent of $K$. By using Vandermonde matrices or AIR matrices, every receiver require to solve $K-D$ equations with $K-D$ unknowns to know its wanted message. The contributions of this paper may be summarized as below:
\begin{itemize}
\item First  we express the optimal scalar linear codes given in \cite{VaR2} in terms of a set of Boolean expressions. 
\item  We analyze some of the combinatorial properties of AIR matrices. By using these properties, we give a simple low-complexity decoding procedure for one-sided SNC-SUICP index codes by using AIR matrices. By using this decoding procedure, every receiver can decode its wanted message by simply adding some subset of index code symbols (broadcast symbols) instead of solving $K-D$ linear equations.
\item We give a reduced set of side-information required by each receiver to decode its wanted message. We explicitly specify this reduced set and the set of broadcast transmissions to be used by every receiver.
\item Our decoding procedure identifies receivers that will perform better than some others when the broadcast channel is noisy. Moreover, it also identifies the receivers which will be able to decode their wanted messages instantly (without using buffers), i.e., the receivers for which the index code appears as an instantly decodable network code (IDNC) \cite{DSAA}.
\end{itemize}

The remaining part of this paper is organized as follows. Definition of AIR matrices by an algorithmic Construction and the structure of their submatrices are given in Section \ref{sec2}. In Section \ref{sec3}, we present the optimal index code for one-sided SNC-SUICP by using AIR matrices in Boolean expressions. The low-complexity  decoding and the reduced side-information patterns used by the receivers are explained in  Section \ref{sec4}. In \ref{sec5}, we study the performance of proposed reduced complexity decoding for one-sided SNC-SUICP in noisy index coding environment. We conclude the paper in Section \ref{sec6} with a brief discussion and possible directions for further work. 

All the subscripts in this paper are to be considered $~\text{\textit{modulo}}~ K$. 
\begin{figure*}
\centering
\includegraphics[scale=0.6]{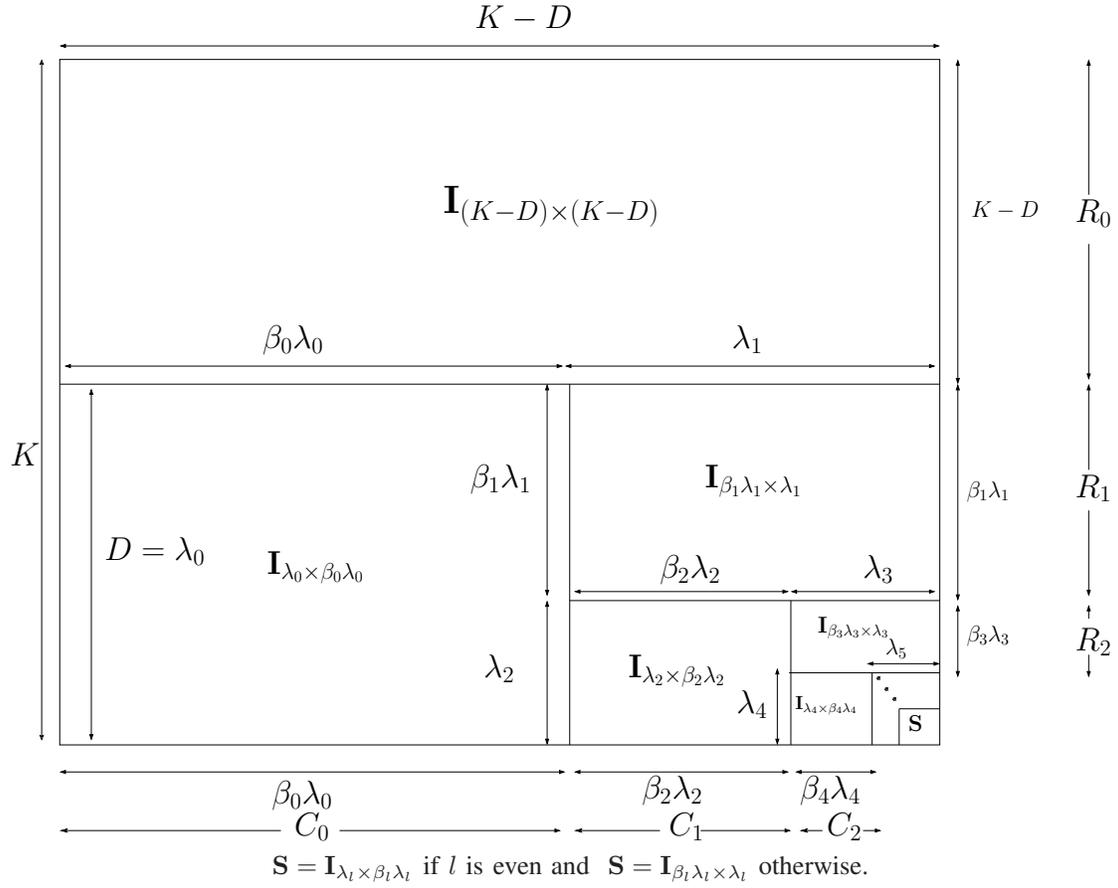}\\
~ $\mathbf{S}=\mathbf{I}_{\lambda_{l} \times \beta_l \lambda_{l}}$ if $l$ is even and ~$\mathbf{S}=\mathbf{I}_{\beta_l\lambda_{l} \times \lambda_{l}}$ otherwise.
\caption{AIR matrix of size $K \times (K-D)$.}
\label{fig1}
~ \\
\hrule
\end{figure*}
\section{AIR matrices and their sub-matrices}
\label{sec2}
In this section, we express optimal scalar linear index code for one-sided SNC-SUICP by using AIR matrices in Boolean expression. We study some of the combinatorial properties of AIR matrices  and by using these properties, we give the simplified decoding and reduced side-information patterns for the SNC-SUICP. 

Given $K$ and $D$ the $K \times (K-D)$  matrix obtained by Algorithm I is called the $(K,D)$ AIR matrix and it is denoted by $\mathbf{L}_{K\times (K-D)}.$ The general form of the $(K,D)$ AIR matrix is shown in   Fig. \ref{fig1}. It consists of several submatrices (rectangular boxes) of different sizes as shown in Fig.\ref{fig1}. The location and sizes of these submatrices are used subsequently to obtain a Boolean expression for the index code and a Boolean expression for the decoding of all the messages. \\
The description of the submatrices are as follows: Let $m$ and $n$ be two positive integers and $n$ divides $m$. The following matrix  denoted by $\mathbf{I}_{m \times n}$ is a rectangular  matrix.
\begin{align}
\label{rcmatrix}
\mathbf{I}_{m \times n}=\left.\left[\begin{array}{*{20}c}
   \mathbf{I}_n  \\
   \mathbf{I}_n  \\
   \vdots  \\
   \mathbf{I}_n 
   \end{array}\right]\right\rbrace \frac{m}{n}~\text{number~of}~ \mathbf{I}_n~\text{matrices}
\end{align}
and $\mathbf{I}_{n \times m}$ is the transpose of $\mathbf{I}_{m \times n}.$ We will call the $\mathbf{I}_{m \times n}$ matrix the $(m \times n)$ identity matrix. 

		\begin{algorithm}
			{Algorithm to construct the AIR matrix $\mathbf{L}$ of size $K \times (K-D)$}
			\begin{algorithmic}[1]
				 \item []Given $K$ and $D$ let $\mathbf{L}=K \times (K-D)$ blank unfilled matrix.
				\item [Step 1]~~~
				\begin{itemize}
				\item[\footnotesize{1.1:}] Let $K=q(K-D)+r$ for $r < K-D$.
				\item[\footnotesize{1.2:}] Use $\mathbf{I}_{q(K-D) \times (K-D)}$ to fill the first $q(K-D)$ rows of the unfilled part of $\mathbf{L}$.
				\item[\footnotesize{1.3:}] If $r=0$,  Go to Step 3.
				\end{itemize}

				\item [Step 2]~~~
				\begin{itemize}
				\item[\footnotesize{2.1:}] Let $(K-D)=q^{\prime}r+r^{\prime}$ for $r^{\prime} < r$.
				\item[\footnotesize{2.2:}] Use $\mathbf{I}_{q^{\prime}r \times r}^{\mathsf{T}}$ to fill the first $q^{\prime}r$ columns of the unfilled part of $\mathbf{L}$.
			    \item[\footnotesize{2.3:}] If $r^{\prime}=0$,  go to Step 3.	
				\item[\footnotesize{2.4:}] $K\leftarrow r$ and $K-D\leftarrow r^{\prime}$.
				\item[\footnotesize{2.5:}] Go to Step 1.
				\end{itemize}
				\item [Step 3] Exit.
		
			\end{algorithmic}
			\label{algo1}
		\end{algorithm}

The simplicity of Algorithm I is illustrated by the following two AIR matrices, discussed in the following section,  $\mathbf{L}_{10 \times 7}$ and $\mathbf{L}_{17 \times 10}$ in Examples \ref{ex1} and \ref{ex2} respectively.  AIR matrices  $\mathbf{L}_{13 \times 10},$  $\mathbf{L}_{13 \times 3}$ and $\mathbf{L}_{44 \times 27}$ can be seen in Examples 3,4 and 6 in the following sections.
\arraycolsep=1.0pt
\setlength\extrarowheight{-2pt}
{
$$\mathbf{L}_{10 \times 7}=\left[
\begin{array}{cccccccccc}
1 & 0 & 0 & 0 & 0 & 0 & 0\\
0 & 1 & 0 & 0 & 0 & 0 & 0\\
0 & 0 & 1 & 0 & 0 & 0 & 0\\
0 & 0 & 0 & 1 & 0 & 0 & 0\\
0 & 0 & 0 & 0 & 1 & 0 & 0\\
0 & 0 & 0 & 0 & 0 & 1 & 0\\
0 & 0 & 0 & 0 & 0 & 0 & 1\\
\hline
1 & 0 & 0 & 1 & 0 & 0 &\vline 1\\
0 & 1 & 0 & 0 & 1 & 0 & \vline 1\\
0 & 0 & 1 & 0 & 0 & 1 &\vline 1\\
 \end{array}
\right]$$
}

\arraycolsep=0.9pt
\setlength\extrarowheight{-2.0pt}
{
$$\mathbf{L}_{17 \times 10}=\left[
\begin{array}{cccccccccc}
1 & 0 & 0 & 0 & 0 & 0 & 0 & 0 & 0 & 0\\
0 & 1 & 0 & 0 & 0 & 0 & 0 & 0 & 0 & 0\\
0 & 0 & 1 & 0 & 0 & 0 & 0 & 0 & 0 & 0\\
0 & 0 & 0 & 1 & 0 & 0 & 0 & 0 & 0 & 0\\
0 & 0 & 0 & 0 & 1 & 0 & 0 & 0 & 0 & 0\\
0 & 0 & 0 & 0 & 0 & 1 & 0 & 0 & 0 & 0\\
0 & 0 & 0 & 0 & 0 & 0 & 1 & 0 & 0 & 0\\
0 & 0 & 0 & 0 & 0 & 0 & 0 & 1 & 0 & 0\\
0 & 0 & 0 & 0 & 0 & 0 & 0 & 0 & 1 & 0\\
0 & 0 & 0 & 0 & 0 & 0 & 0 & 0 & 0 & 1\\
\hline
1 & 0 & 0 & 0 & 0 & 0 & 0 &\vline 1 & 0 & 0\\
0 & 1 & 0 & 0 & 0 & 0 & 0 &\vline 0 & 1 & 0\\
0 & 0 & 1 & 0 & 0 & 0 & 0 &\vline 0 & 0 & 1\\
0 & 0 & 0 & 1 & 0 & 0 & 0 &\vline 1 & 0 & 0\\
0 & 0 & 0 & 0 & 1 & 0 & 0 &\vline 0 & 1 & 0\\
0 & 0 & 0 & 0 & 0 & 1 & 0 &\vline \underline{0} & \underline{0} & \underline{1}\\
0 & 0 & 0 & 0 & 0 & 0 & 1 &\vline 1 & 1 & 1\\
 \end{array}
\right]$$
}

AIR matrices have been shown to be giving optimal length index codes for one-sided SNC-SUICP problems \cite{VaR2} and are shown to be useful to obtain optimal length index codes for two-sided SNC-SUICP problems in \cite{VaR1}. 

Towards explaining the other quantities shown in the AIR matrix shown in Fig. \ref{fig1}, for a given $K$  and $D,$ let  $D= \lambda_0$ and\begin{align}
\nonumber
K-D&=\beta_0 \lambda_0+\lambda_1, \nonumber \\
\lambda_0&=\beta_1\lambda_1+\lambda_2, \nonumber \\
\lambda_1&=\beta_2\lambda_2+\lambda_3, \nonumber \\
\lambda_2&=\beta_3\lambda_3+\lambda_4, \nonumber \\
&~~~~~~\vdots \nonumber \\
\lambda_i&=\beta_{i+1}\lambda_{i+1}+\lambda_{i+2}, \nonumber \\ 
&~~~~~~\vdots \nonumber \\ 
\lambda_{l-1}&=\beta_l\lambda_l.
\label{chain}
\end{align}
where $\lambda_{l+1}=0$ for some integer $l,$ $\lambda_i,\beta_i$ are positive integers and $\lambda_i < \lambda_{i-1}$ for $i=1,2,\ldots,l$. The number of submatrices in the AIR matrix is $l+2$ and the size of each submatrix is shown using $\lambda_i,\beta_i,$  $i=0,1,2,\ldots,l.$ The submatrices are classified in to the following three types. 
\begin{itemize}
\item The first submatrix is the $I_{(K-D) \times (K-D)}$ matrix at the top of Fig. \ref{fig1} which is independent of $\lambda_i,\beta_i,$  $i=0,1,2,\ldots,l.$ This will be referred as the $I_{(K-D) \times (K-D)}$ matrix henceforth.
\item The set of matrices of the form $I_{\lambda_i \times \beta_i \lambda_i}$ for $i=0,2,4, \cdots$ (for all $i$ even) will be referred as the set of even-submatrices.
\item The set of matrices of the form $I_{\beta_i \lambda_i \times  \lambda_i}$ for $i=1,3,5, \cdots$ (for all $i$ odd) will be referred as the set of odd-submatrices. 
\end{itemize}
Note that the odd-submatrices are always "fat" and the even-submatrices are always "tall" including square matrices in both the sets. 
By the $i$-th submatrix is meant either an odd-submatrix or an even-submatrix for $ 0 \leq i \leq l.$  Also whenever  $\beta_0=0,$  the corresponding submatrix will not exist in the AIR matrix.  
To  express the optimal scalar linear index code given by an AIR matrix in terms of Boolean expressions as well as to present a low complexity decoding it is required in the following  sections that the location of both the odd- and even-submatrices within the AIR matrix need to be identified. Towards this end, we define the following intervals. Let $R_0,R_1,R_2,\ldots,R_{\left\lfloor \frac{l}{2}\right\rfloor}+1$ be the intervals that will identify the rows of the submatrices  as given below:
\begin{itemize}
\item $R_0=[0:K-\lambda_0-1]$
\item $R_1=[K-\lambda_0:K-\lambda_2-1]$
\item $R_2=[K-\lambda_2:K-\lambda_4-1]$
\item []~~~~~~~~~~~~~~$\vdots$
\item $R_i=[K-\lambda_{2(i-1)}:K-\lambda_{2i}-1]$
\item []~~~~~~~~~~~~~~$\vdots$
\item $R_{\left\lfloor
\frac{l}{2}\right\rfloor}=[K-\lambda_{2(\left\lfloor
\frac{l}{2}\right\rfloor-1)}:K-\lambda_{2\left\lfloor
\frac{l}{2}\right\rfloor}-1]$
\item $R_{\left\lfloor \frac{l}{2}\right\rfloor+1}=[K-\lambda_{2\left\lfloor \frac{l}{2}\right\rfloor}:K-1]$,
\end{itemize} 
we have $R_0\cup R_1 \cup R_2 \cup \ldots \cup R_{\left\lfloor \frac{l}{2}\right\rfloor +1}=[0:K-1]$.

Let $C_0,C_1,\ldots,C_{\left\lceil \frac{l}{2}\right\rceil}$ be the intervals that will identify the columns of the submatrices  as given below:
\begin{itemize}

\item $C_0=[0:\beta_0\lambda_0-1]$ if $\beta_0 \geq 1$, else $C_0=\phi$
\item $C_1=[K-D-\lambda_1:K-D-\lambda_3-1]$
\item $C_2=[K-D-\lambda_3:K-D-\lambda_5-1]$
\item [] ~~~~~~~~~~~$\vdots$
\item $C_i=[K-D-\lambda_{2i-1}:K-D-\lambda_{2i+1}-1]$
\item []~~~~~~~~~~~~~~$\vdots$
\item $C_{\left\lceil
\frac{l}{2}\right\rceil-1}=[K-D-\lambda_{2\left\lceil
\frac{l}{2}\right\rceil-3}:K-D-\lambda_{2\left\lceil
\frac{l}{2}\right\rceil-1}-1]$
\item $C_{\left\lceil \frac{l}{2}\right\rceil}=[K-D-\lambda_{2\left\lceil \frac{l}{2}\right\rceil-1}:K-D-1]$
\end{itemize} 
\begin{align}
\label{column}
\end{align}
we have $C_0\cup C_1 \cup C_2 \cup \ldots \cup C_{\left\lceil \frac{l}{2}\right\rceil}=[0:K-D-1]$.

\section{Boolean expressions for optimal length index codes for SNC-SUICP}
\label{sec3}

Let $\mathbf{L}$ be a AIR matrix of size $K \times (K-D)$. A scalar linear index code of length $K-D$ generated by an AIR matrix is given by
\begin{align}
\label{code1}
 [c_0~c_1~\ldots~c_{K-D-1}]=[x_0~x_1~\ldots~x_{K-1}]\mathbf{L}=\sum_{k=0}^{K-1}x_kL_k
\end{align}
where $L_k$ is the $k$th row of $\mathbf{L}$ for $k\in[0:K-1]$. In the matrix $\mathbf{L}$, the $j$th column contains the coefficients used for mixing messages $x_0,x_1,\ldots,x_{K-1}$ to get the $j$th broadcast symbol.
\begin{theorem}
\label{thm1}
The broadcast symbol generated by the $k$th column of an AIR matrix for $k \in [0:K-D-1]$ can be expressed in Boolean expression as given at the top of the next page where $I_{S}$ denotes the indicator function, i.e., $I_{{S}}$ takes the value $1$ if the statement $S$ is true and $0$ otherwise. We will refer to this Boolean expression as the Boolean expression for $c_k.$ 
\end{theorem}
\begin{proof}
Proof is given in Appendix A.
\end{proof}
\begin{figure*}
\begin{align*}
\nonumber
c_k=x_k&+I_{\{k \in C_0\}}(x_{(K-D)+k~\text{\textit{mod}}~\lambda_0})
\\&
+I_{\{k \in C_1\cup C_2 \cup \ldots \cup C_{\left\lceil \frac{l}{2}\right\rceil}\}}\bigg(\sum_{j=1}^{\beta_1} x_{k+j\lambda_1}\bigg)
+I_{\{k \in C_1\}}(x_{ (K-\lambda_2)+(k-K+D+\lambda_1)~\text{\textit{mod}}~\lambda_2}) \\&
\nonumber
+I_{\{k \in C_2 \cup C_3 \cup \ldots \cup C_{\left\lceil \frac{l}{2}\right\rceil}\}}\bigg(\sum_{j=1}^{\beta_3} x_{k+\beta_1\lambda_1+j\lambda_3}\bigg)
+I_{\{k \in C_2\}}(x_{(K-\lambda_4)+(k-K+D+\lambda_3)~\text{\textit{mod}}~\lambda_4 }) \\&
\nonumber
+I_{\{k \in C_3 \cup \ldots \cup C_{\left\lceil \frac{l}{2}\right\rceil}\}}\bigg(\sum_{j=1}^{\beta_5} x_{k+\beta_1\lambda_1+\beta_3\lambda_3+j\lambda_5}\bigg)
+I_{\{k \in C_3\}}(x_{(K-\lambda_6)+(k-K+D+\lambda_5)~\text{\textit{mod}}~\lambda_6 })\\&
\nonumber
+\ldots \\&
\nonumber
+I_{\{k \in C_i \cup C_{i+1} \cup \ldots \cup C_{\left\lceil \frac{l}{2}\right\rceil}\}}
\bigg(\sum_{j=1}^{\beta_{2i-1}} x_{k+\beta_1\lambda_1+\beta_3\lambda_3+\ldots+\beta_{2i-3}\lambda_{2i-3}+j\lambda_{2i-1}}\bigg)
+I_{\{k \in C_i\}}(x_{(K-\lambda_{2i})+(k-K+D+\lambda_{2i-1})~\text{\textit{mod}}~\lambda_{2i} })\\&
\nonumber
+\ldots\\&
\nonumber
+I_{\{k \in C_{\left\lceil \frac{l}{2}\right\rceil-1}\}} 
(x_{(K-\lambda_{2\left\lceil \frac{l}{2}\right\rceil-2})+(k-K+D+\lambda_{2\left\lceil \frac{l}{2}\right\rceil-3})~\text{\textit{mod}}~\lambda_{2(\left\lceil \frac{l}{2}\right\rceil-1)} }) \\&
+I_{\{k \in C_{\left\lceil \frac{l}{2}\right\rceil}\}} 
\nonumber
\bigg(\sum_{j=1}^{\beta_{2\left\lceil \frac{l}{2}\right\rceil-1}} x_{k+\beta_1\lambda_1+\ldots+\beta_{2\left\lceil \frac{l}{2}\right\rceil-3}\lambda_{2\left\lceil \frac{l}{2}\right\rceil-3}+j\lambda_{2\left\lceil \frac{l}{2}\right\rceil-1}}\bigg)\\&
+I_{\{l ~ is ~ even\}} I_{\{k \in C_{\left\lceil \frac{l}{2}\right\rceil}\}}(x_{(K-\lambda_l)+(k-K+D+\lambda_{l-1})~\text{\textit{mod}}~\lambda_l }) 
\end{align*}
\begin{equation}
\label{code}
\hspace*{-6.5cm}+I_{\{l~ is ~odd\}}I_{\{k \in C_{\left\lceil \frac{l}{2}\right\rceil}\}} 
\bigg(\sum_{j=1}^{\beta_{l}} x_{k+\beta_1\lambda_1+\beta_3\lambda_3+\ldots+\beta_{l-2}\lambda_{l-2}+
j\lambda_{l}}\bigg)
\end{equation}

\hrule
\end{figure*} 
\begin{example}
\label{ex1}
For the cases where  $l=1,$  the scalar linear index code is given by 
\begin{align*}
\nonumber
c_k=x_k&+I_{\{k \in C_0\}}(x_{k~\text{\textit{mod}}~\lambda_0)+(K-D)})\\&+I_{\{k \in C_1\}}\bigg(\sum_{j=1}^{\beta_1} x_{k+j\lambda_1}\bigg). 
\end{align*}
For example consider the case  $K=10$ and $D=3$. For this SNC-SUICP, we have $\lambda_0=3,\beta_0=2,\lambda_1=1,l=1,C_0=[0:5]$ and $C_1=\{6\}$. The scalar linear index code is given by 
\begin{align*}
\mathfrak{C}=\{&x_0+x_7,~~~x_1+x_8,~~~x_2+x_{9},~~~x_3+x_7,\\&x_4+x_8,~~~x_5+x_{9},
~~~x_6+x_7+x_8+x_{9}\}.
\end{align*}
The encoding matrix $\mathbf{L}_{10 \times 7}$ was shown in the previous section.
\end{example}
\begin{example}
\label{ex2}
For the cases leading to  $l=2,$ the scalar linear index code is given by 
\begin{align*}
\nonumber
c_k=x_k&+I_{\{k \in C_0\}}(x_{k~\text{\textit{mod}}~\lambda_0+(K-D)})\\&+I_{\{k \in C_1\}}\bigg(\sum_{j=1}^{\beta_1} x_{k+j\lambda_1}\bigg)\\&+I_{\{k \in C_1\}}(x_{(k-K+D+\lambda_1)~\text{\textit{mod}}~(\lambda_2) + K-\lambda_2}).
\end{align*}

Consider the example  $K=17$ and $D=7$ which leads to $l=2.$. For this SNC-SUICP, we have $\lambda_0=7,\beta_0=1,\lambda_1=3,\beta_1=2,\lambda_2=1,C_0=[0:6]$ and $C_1=[7:9]$. The encoding matrix $\mathbf{L}_{17 \times 10}$ was shown in the previous section. The scalar linear index code in Boolean expressions is given by 
\begin{align*}
\mathfrak{C}=\{&x_0+x_{10},~~~x_1+x_{11},~~~x_2+x_{12},\\&x_3+x_{13},~~~x_4+x_{14},~~~x_5+x_{15},
\\&x_6+x_{16},~~~x_7+x_{10}+x_{13}+x_{16},\\&x_8+x_{11}+x_{14}+x_{16},~~~x_{9}+x_{12}+x_{15}+x_{16}\}.
\end{align*}
\end{example}
\section{Low Complexity Decoding of SNC-SUICP using Boolean Expressions}
\label{sec4}

In this section we present a low complexity decoding of SNC-SUICP using the Boolean expression description of the code obtained in the previous section. To be specific, we give closed form expressions for the message symbols of  every receiver as a function of broadcast symbols.  

\subsection{Properties of AIR matrices}
In this subsection,  we define some quantities related to an  AIR matrix that are used subsequently to obtain the closed form expression for decoding.  Let $\mathbf{L}$ be the AIR matrix of size $K \times (K-D)$. 

In the matrix $\mathbf{L}$, the element $\mathbf{L}(j,k)$ is present in one of the submatrices: $\mathbf{I}_{K-D}$ or $\mathbf{I}_{\beta_{2i+1}\lambda_{2i+1} \times \lambda_{2i+1}}$ for $i \in [0:\lceil \frac{l}{2}\rceil-1]$ or $\mathbf{I}_{\lambda_{2i} \times \beta_{2i}\lambda_{2i}}$ for $i \in [0:\left\lfloor\frac{l}{2}\right\rfloor]$. Let $(j_R,k_R)$ be the (row-column) indices of $\mathbf{L}(j,k)$ within the submatrix in which $\mathbf{L}(j,k)$ is present. Then, for a given $\mathbf{L}(j,k)$, the indices $j_R$ and $k_R$ are as given below.
\begin{itemize}
\item If $\mathbf{L}(j,k)$ is present in $\mathbf{I}_{K-D}$, then \\  $j_R=j$ and $k_R=k$.
\item If $\mathbf{L}(j,k)$ is present in $\mathbf{I}_{\lambda_{0} \times \beta_{0}\lambda_{0}}$, then \\
 $j_R=j~\text{\textit{mod}}~(K-D)$ and $k_R=k$. 
\item If $L(j,k)$ is present in $\mathbf{I}_{\beta_{2i+1}\lambda_{2i+1} \times \lambda_{2i+1}}$ for $i \in [0:\lceil \frac{l}{2}\rceil-1]$, then \\ $j_R=j~\text{\textit{mod}}~(K-\lambda_{2i})$ and $k_R=k~\text{\textit{mod}}~(K-D-\lambda_{2i+1})$.
\item If $L(j,k)$ is present in $\mathbf{I}_{\lambda_{2i} \times \beta_{2i}\lambda_{2i}}$ for $i \in [1:\left\lfloor\frac{l}{2}\right\rfloor]$, then $j_R=j~\text{\textit{mod}}~(K-\lambda_{2i})$ and $k_R=k~\text{\textit{mod}}~(K-D-\lambda_{2i-1})$. 
\end{itemize}

In Definition \ref{def1} below we define several distances between the $1$s present in an AIR matrix. These distances are used  in the following subsection to obtain a low complexity decoding.  Figure \ref{sfig4} is useful to visualize the distances defined.

\begin{definition}
\label{def1}
Let $\mathbf{L}$ be the AIR matrix of size $K \times (K-D)$.
\begin{itemize}
\item [\textbf{(i)}] For $k\in [0:K-D-1]$ we have  $\mathbf{L}(k,k)=1.$  Let $k^{\prime}$ be the maximum integer such that $k^{\prime} > k$ and $\mathbf{L}(k^{\prime},k)=1$. Then $k^{\prime}-k,$ denoted by $d_{down}(k),$  is called the down-distance of $\mathbf{L}(k,k)$.  
\item [\textbf{(ii)}] Let $\mathbf{L}(j,k)=1$ and $j \geq K-D$. Let $j^{\prime}$ be the maximum integer such that $j^{\prime} < j$ and $\mathbf{L}(j^{\prime},k)=1$. Then $j-j^{\prime},$ denoted by $d_{up}(j,k),$ is called the up-distance of $\mathbf{L}(j,k).$ 

\item [\textbf{(iii)}]  Let $\mathbf{L}(j,k)=1$ and $\mathbf{L}(j,k)\in \mathbf{I}_{ \lambda_{2i} \times \beta_{2i} \lambda_{2i}}$ for $i \in [0:\lfloor \frac{l}{2}\rfloor]$. Let $k^{\prime}$ be the minimum integer such that $k^{\prime} > k$ and $\mathbf{L}(j,k^{\prime})=1$. Then $k^{\prime}-k,$ denoted by $d_{right}(j,k),$  is called the right-distance of $\mathbf{L}(j,k).$  

\item [\textbf{(iv)}] For $k\in [0:K-D-\lambda_l-1]$, let $d_{right}(k+d_{down}(k),k)=\mu_k.$  Let the number of $1$s in the $(k+\mu_k)$th column of $\mathbf{L}$ below $\mathbf{L}(k+d_{down}(k),k+\mu_k)$ be $p_k$ and these are  at a distance of $t_{k,1},t_{k,2},\ldots,t_{k,p_k}~(t_{k,1}<t_{k,2}<\ldots <t_{k,p_k})$ from $\mathbf{L}(k+d_{down}(k),k+\mu_k)$. Then, $t_{k,r}$ is called the $r$-th  down-distance of $\mathbf{L}(k+d_{down}(k),k+\mu_k)$, for $1 \leq r \leq p_k.$
\end{itemize}
\end{definition}

Notice that if $L(k+d_{down}(k),k+\mu_k) \in \mathbf{I}_{\lambda_{2i} \times   \beta_{2i} \lambda_{2i}}$ in $\mathbf{L}$ for $i=0,1,2,\ldots,\lfloor \frac{l}{2}\rfloor$, then $p_k=0$.

\begin{figure}
\centering
\includegraphics[scale=0.62]{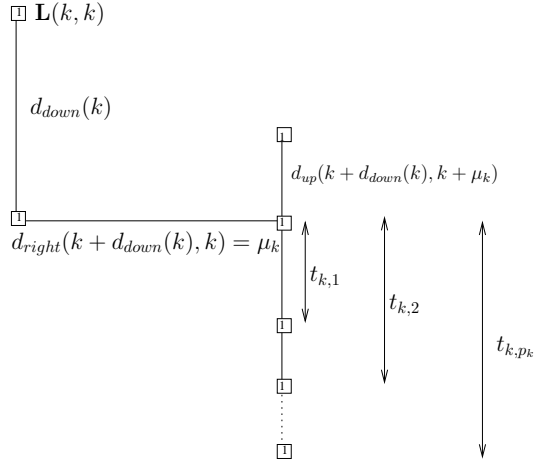}\\
\caption{Illustration of Definition \ref{def1}}
\label{sfig4}
\end{figure}

\begin{lemma}
\label{lemma1}
Let $k \in C_i$ for $i \in [0:\lceil\frac{l}{2}\rceil]$. Let $k~\text{mod}~(K-D-\lambda_{2i-1})=c\lambda_{2i}+d$ for some positive integers $c$ and $d$ $(d<\lambda_{2i})$. The down distance is given by
\begin{align}
\label{mdd}
d_{down}(k)=D+\lambda_{2i+1}+(\beta_{2i}-1-c)\lambda_{2i}.
\end{align}
\end{lemma}
\begin{proof}
Proof is given in Appendix B.
\end{proof}

\begin{lemma}
\label{lemma2}
The up-distance of $\mathbf{L}(j,k)$ is as given below.
\begin{itemize}
\item If $\mathbf{L}(j,k) \in \mathbf{I}_{\beta_{2i+1} \lambda_{2i+1} \times \lambda_{2i+1}}$  for $i \in [0:\lceil \frac{l}{2}\rceil-1]$, then $d_{up}(j,k)$ is $\lambda_{2i+1}$. 
\item If $\mathbf{L}(j,k) \in \mathbf{I}_{\lambda_{2i} \times  \beta_{2i} \lambda_{2i}}$  for $i\in [0:\lfloor \frac{l}{2}\rfloor]$ and $k_R=c\lambda_{2i}+d$ for some positive integer $c,d~(d < \lambda_{2i})$, then $d_{up}(j,k)$ is $\lambda_{2i-1}-c\lambda_{2i}$. 
\end{itemize}
\end{lemma}
\begin{proof}
Proof is given in Appendix C.
\end{proof}

\begin{lemma}
\label{lemma3}
The right-distance of $\mathbf{L}(j,k)$ is as given below.
\begin{itemize}
\item If $k_R \in [0:(\beta_{2i}-1)\lambda_{2i}-1]$ for $i \in [0:\lfloor \frac{l}{2}\rfloor]$, then $d_{right}(j,k)$ is $\lambda_{2i}$. 
\item If $k_R \in [(\beta_{2i}-1)\lambda_{2i}:\beta_{2i} \lambda_{2i}-1]$ for $i \in [0:\lfloor \frac{l}{2}\rfloor-1]$, then $d_{right}(j,k)$ depends on $j_R$. If $j_R=c\lambda_{2i+1}+d$ for some positive integers $c,d~(d<\lambda_{2i+1})$, then  $d_{right}(j,k)$ is $\lambda_{2i}-c\lambda_{2i+1}$. 
\end{itemize}
\end{lemma}
\begin{proof}
Proof is given in Appendix D.
\end{proof}
\begin{example}
\label{ex3}
Consider the AIR matrix of size $13 \times 10$ given shown below. For this matrix  $D=3, ,\beta_0=3,\lambda_1=1,\beta_1=3$ and $l=1$.  

\arraycolsep=1.0pt
\setlength\extrarowheight{-2pt}
{
$$\mathbf{L}_{13 \times 10}=\left[
\begin{array}{cccccccccc}
{\color{red}\textbf{1}} & 0 & 0 & 0 & 0 & 0 & 0 & 0 & 0 & 0\\
0 & {\color{red}\textbf{1}} & 0 & 0 & 0 & 0 & 0 & 0 & 0 & 0\\
0 & 0 & {\color{red}\textbf{1}} & 0 & 0 & 0 & 0 & 0 & 0 & 0\\
0 & 0 & 0 & {\color{green}\textbf{1}} & 0 & 0 & 0 & 0 & 0 & 0\\
0 & 0 & 0 & 0 & {\color{green}\textbf{1}} & 0 & 0 & 0 & 0 & 0\\
0 & 0 & 0 & 0 & 0 & {\color{green}\textbf{1}} & 0 & 0 & 0 & 0\\
0 & 0 & 0 & 0 & 0 & 0 & {\color{blue}\textbf{1}} & 0 & 0 & 0\\
0 & 0 & 0 & 0 & 0 & 0 & 0 & {\color{blue}\textbf{1}} & 0 & 0\\
0 & 0 & 0 & 0 & 0 & 0 & 0 & 0 & {\color{blue}\textbf{1}} & 0\\
0 & 0 & 0 & 0 & 0 & 0 & 0 & 0 & 0 & {\color{red}\textbf{1}}\\
\hline
{\color{red}\textbf{1}} & 0 & 0 & {\color{green}\textbf{1}} & 0 & 0 & {\color{blue}\textbf{1}} & 0 & 0 &\vline 1\\
0 & {\color{red}\textbf{1}} & 0 & 0 & {\color{green}\textbf{1}} & 0 & 0 & {\color{blue}\textbf{1}} & 0 &\vline 1\\
0 & 0 & {\color{red}\textbf{1}} & 0 & 0 & {\color{green}\textbf{1}} & 0 & 0 & {\color{blue}\textbf{1}} &\vline {\color{red}\textbf{1}}\\
 \end{array}
\right]$$
}

The  down-distances of the elements in the matrix $\mathbf{L}_{13 \times 10}$ are given below.
\begin{align*}
&d_{down}(0)=d_{down}(1)=d_{down}(2)=10,\\&
d_{down}(3)=7,~d_{down}(4)=7,~~d_{down}(5)=7,\\&
d_{down}(6)=4,~d_{down}(7)=4,~~d_{down}(8)=4,\\&
d_{down}(9)=3.
\end{align*}

The up-distances are given below.
\begin{align*}
&d_{up}(10,0)=d_{up}(11,1)=d_{up}(12,2)=9,\\&
d_{up}(10,3)=d_{up}(11,4)=d_{up}(12,5)=6,\\&
d_{up}(10,6)=d_{up}(11,7)=d_{up}(12,8)=3,\\&
d_{up}(10,9)=d_{up}(11,9)=d_{up}(12,9)=1.
\end{align*}

The right-distances are given below.
\begin{align*}
&d_{right}(10,0)=3,d_{right}(10,3)=3,d_{right}(10,6)=3,\\&
d_{right}(11,1)=3,d_{right}(11,4)=3,d_{right}(11,7)=2,\\&
d_{right}(12,2)=3,d_{right}(12,5)=3,d_{right}(12,8)=1.
\end{align*}

We have $p_k=0~\text{for}~k=0,1,2,3,4,5,8,9.$ For $k=6$, we have $p_k=2,t_{6,1}=1$ and $t_{6,2}=2$. For $k=7$, we have $p_k=1$ and $t_{7,1}=1.$
\end{example}
\begin{example}
\label{ex4}
Consider the AIR matrix of size $13 \times 3$ given below. In this matrix $K=13, D=10, K-D=3,\beta_0=0,\lambda_1=3,\beta_1=3,\lambda_2=1,\beta_2=3$ and $l=2$. 
\arraycolsep=1.0pt
\setlength\extrarowheight{-2pt}
{
$$\mathbf{L}_{13 \times 3}=\left[
\begin{array}{cccccccccc}
1 & 0 & 0\\
0 & 1 & 0\\
0 & 0 & 1\\
\hline
1 & 0 & 0\\
0 & 1 & 0\\
0 & 0 & 1\\
1 & 0 & 0\\
0 & 1 & 0\\
0 & 0 & 1\\
1 & 0 & 0\\
0 & 1 & 0\\
0 & 0 & 1\\
\hline
1 & 1 & 1\\
 \end{array}
\right]$$
}

The down-distances are given below.
\begin{align*}
d_{down}(0)=12,d_{down}(1)=11,d_{down}(2)=10.
\end{align*}

The up-distances are given below.
\begin{align*}
&d_{up}(3,0)=d_{up}(4,1)=d_{up}(5,2)=3,\\&
d_{up}(6,0)=d_{up}(7,1)=d_{up}(8,2)=3,\\&
d_{up}(9,0)=d_{up}(10,1)=d_{up}(11,2)=3,\\&
d_{up}(12,0)=3,d_{up}(12,1)=2,d_{up}(12,2)=1.
\end{align*}

The right-distances given below.
\begin{align*}
d_{right}(12,0)=d_{right}(12,1)=1.
\end{align*}

We have $p_k=0$ for $k \in [0:12]$.
\end{example}

\subsection{Low Complexity decoding for One-Sided SNC-SUICP}
In this subsection, we give a low complexity decoding procedure for one-sided SNC-SUICP index codes by using AIR matrices. 
Henceforth, for the sake of notational convenience we denote $K-D$ by $\lambda_{-1}$ also.

It turns out that the interval $C_{i}$ defined in \eqref{column} for $i=0,1,2,\ldots,\left\lceil\frac{l}{2}\right\rceil$
needs to be partitioned into two  as $C_{i} =D_{i}\cup E_{i}$ as given below for the purpose of obtaining expressions for decoding for the receivers with indices falling in $C_i.$ This is because the right-distances for entries in the submatrices corresponding to $C_i$ are different for $D_i$ and $E_i.$    
\begin{align}
\label{setd}
\nonumber
&D_{i}=[K-D-\lambda_{2i-1}:K-D-\lambda_{2i-1}+(\beta_{2i}-1)\lambda_{2i}-1]\\&
E_i=[K-D-\lambda_{2i-1}+(\beta_{2i}-1)\lambda_{2i}:K-D-\lambda_{2i+1}-1].
\end{align}

\begin{theorem}
\label{thm2}
Let $\mathbf{L}$ be the AIR matrix of size $K \times (K-D)$. Let $[c_0~c_1~\ldots~c_{K-D-1}]$ be the optimal scalar linear index code generated by $\mathbf{L}$ for SNC-SUICP. Let $d_{right}(k+d_{down}(k),k)=\mu_k.$ Also let $t_{k,r}$ be the $r$-th  down distance of $\mathbf{L}(k+d_{down}(k),k+\mu_k),$ where $1 \leq r \leq p_k.$  Let $D_i$ and $E_i$ be the sets as given in \eqref{setd}. Then, the receiver $R_k$ can decode its wanted message symbol $x_k$ as given below.
\begin{itemize}
\item [\textbf{(i)}] If $k \in D_i$ for $i \in [0:\lceil \frac{l}{2}\rceil]$, then $R_k$ can decode $x_k$ by using the two broadcast symbols $c_k, c_{k+\mu_k}$ and the available side-information. 
That is 
\begin{align}
\label{case1}
x_k=&c_k+c_{k+\mu_k}+\underbrace{\nu_k+\nu_{k+\mu_k}}_{\text{side-information~terms}},
\end{align}
where 
\begin{align}
\label{case11}
\nu_k=\sum_{z=1}^{i}\sum_{j=1}^{\beta_{2z-1}}x_{k+D-\lambda_{2z-2}+j\lambda_{2z-1}},
\end{align}
and
\begin{align}
\label{case12}
\nu_{k+\mu_k}=x_{k+\mu_k}+\sum_{z=1}^{i}\sum_{j=1}^{\beta_{2z-1}}x_{k+t_{k,j}+D-\lambda_{2z-2}+j\lambda_{2z-1}}.
\end{align}
\item [\textbf{(ii)}] If $k \in E_i$ for $i \in [0:\lceil \frac{l}{2}\rceil-1]$, then $R_k$ can decode $x_k$ by using the broadcast symbols $c_k, c_{k+\mu_k},c_{k+t_{k,1}},\ldots,c_{k+t_{k,p_k}}$ and the side-information available. That is 
\begin{align}
\label{case2}
\nonumber
x_k=&c_k+c_{k+\mu_k}+c_{k+t_{k,1}}+c_{k+t_{k,2}}+\ldots+c_{k+t_{k,p_k}}\\&+\underbrace{\nu_k+\nu_{k+\mu_k}+\nu_{k+t_{k,1}}+\nu_{t_{k,2}}+\ldots+
\nu_{k+t_{k,p_k}}}_{\text{side-information~terms}},
\end{align}
where 
\begin{align}
\label{case21}
\nonumber
&\nu_k=\sum_{z=1}^{i}\sum_{j=1}^{\beta_{2z-1}}x_{k+D-\lambda_{2z-2}+j\lambda_{2z-1}},\\&
\nonumber
\nu_{k+t_{k,j}}=x_{k+t_{k,j}}+\sum_{z=1}^{i}\sum_{j=1}^{\beta_{2z-1}}x_{k+t_{k,j}+D-\lambda_{2z-2}+j\lambda_{2z-1}}~\\&
\nonumber
\text{for}~j\in[1:p_k], \text{and} \\&
\nonumber
\nu_{k+\mu_k}=x_{k+\mu_k}+\sum_{z=1}^{i}\sum_{j=1}^{\beta_{2z-1}}x_{k+t_{k,j}+D-\lambda_{2z-2}+j\lambda_{2z-1}}\\&
\nonumber
+\sum_{j=1}^{c}x_{k+t_{k,j}+D-\lambda_{2i}+j\lambda_{2i+1}}~\text{and}~c\in \mathbb{Z}^+ \text{such~that}\\& ~k_R=(\beta_{2i}-1)\lambda_{2i}+c\lambda_{2i+1}+d~\text{and}~d<\lambda_{2i+1}.
\end{align}
\item [\textbf{(iii)}] If $k \in [K-D-\lambda_l:K-D-1]$, then $R_k$ can decode $x_k$ from the broadcast symbol $c_k$ and the side-information. That is 
\begin{align}
\label{case3}
x_k=&c_k+\nu_k,
\end{align}
where the side-information term
\begin{align*}
&\nu_k=\sum_{z=1}^{\lceil \frac{l}{2}\rceil}\sum_{j=1}^{\beta_{2z-1}}x_{k+D-\lambda_{2z-2}+j\lambda_{2z-1}}~\text{if~$l$~is~odd, and}\\&
\nu_k=x_{k+D}+\sum_{z=1}^{\lceil \frac{l}{2}\rceil}\sum_{j=1}^{\beta_{2z-1}}x_{k+D-\lambda_{2z-2}+j\lambda_{2z-1}}~\text{if~$l$~is~even}.
\end{align*}
\item [\textbf{(iv)}] If $k \in [K-D:K-1]$, then $R_k$ can decode $x_k$ from the broadcast symbol $c_{k~\text{mod}~(K-D)}$. Let $k~\text{mod}~(K-D)=k^{\prime}$. Receiver $R_k$ can decode $x_k$ as 
\begin{align}
\label{case4}
 x_k=c_{k^{\prime}}+\nu_{k^{\prime}},
 \end{align}
where the side-information term $\nu_{k^{\prime}}$ is the exclusive OR of all message symbols present in $c_{k^{\prime}}$ excluding $x_k$.

\end{itemize}
\end{theorem}
\begin{proof}
Proof is given in Appendix E.
\end{proof}
\begin{corollary}
\label{cor1}
Consider a one-sided SNC-SUICP with $K$ messages and $D$ number of side-information. Let $\gamma_k$ be subset of side-information $\mathcal{K}_k=\{x_{k+1},x_{k+2},\ldots,x_{k+D}\}$ used by receiver $R_k$ to decode its wanted message by using the scalar linear index code given by the AIR matrix of size $K \times (K-D)$. $\gamma_k$ for $k \in [0:K-1]$ has been explicitly given in Theorem \ref{thm2}. Consider a new index coding problem with $K$ messages and $K$ receivers with the receiver $R_k$ wanting the message $x_k$ and having side-information $\mathcal{K}_k=\gamma_k$. We call this index coding problem as $(K,D, \{ \gamma_k \}_{k=1}^{K})$ non-symmetric index coding problem. From Theorem \ref{thm2} it follows that the capacity of this non-symmetric side-information index coding problem is $\frac{1}{K-D}$ and the minrank of its side-information graph is $K-D$.
\end{corollary}

\begin{example}
\label{ex5}
Consider a one-sided SNC-SUICP with $K=13, D=3$. For this index coding problem, we have $\lambda_1=1,\beta_0=3,\beta_1=3$ and $l=1$. The encoding matrix $\mathbf{L}_{13 \times 10}$ is given in Example \ref{ex3}. Table \ref{table2} gives the broadcast symbols and side-information set $\gamma_k$ used by receiver $R_k$ to decode $x_k$ for $k\in [0:12]$.
\begin{table*}
\centering
\setlength\extrarowheight{0pt}
\begin{tabular}{|c|c|c|c|c|c|c|c|}
\hline
\textbf{$R_k$} &$\mathcal{W}_k$&$d_{down}(k)$&$\mu_k$&$p_k$& $t_{k,\tau}$&$c_k,c_{k+t_{k,\tau}},c_{k+\mu_k}$&$\gamma_k$\\
\cline{6-7}
&&&&& \multicolumn{1}{c}{}&$\tau=1,2,\ldots,p_k$&~\\
\hline
$R_0$ & $x_0$ & 10&3&0&0&$c_0,c_3$&$x_3$ \\
\hline
$R_1$ & $x_1$ & 10&3&0&0&$c_1,c_4$&$x_4$\\
\hline
$R_2$ & $x_2$ & 10&3&0&0&$c_2,c_5$&$x_5$\\
\hline
$R_3$ & $x_3$ & 7&3&0&0&$c_3,c_6$&$x_6$ \\
\hline
$R_4$ & $x_4$ & 7&3&0&0&$c_4,c_7$&$x_7$ \\
\hline
$R_5$ & $x_5$ & 7&3&0&0&$c_5,c_8$&$x_8$ \\
\hline
$R_6$ & $x_6$ & 4&3&2&1,2&$c_6,c_7,c_8,c_9$&$x_7,x_8,x_9$ \\
\hline
$R_7$ & $x_7$ & 4&2&1&1&$c_7,c_8,c_9$&$x_8,x_9,x_{10}$ \\
\hline
$R_8$ & $x_8$ &4&1&0&0&$c_8,c_9$&$x_9,x_{10},x_{11}$ \\
\hline
$R_{9}$ & $x_{9}$ & 3&-&0&0&$c_9$&$x_{10},x_{11},x_{12}$ \\
\hline
$R_{10}$ & $x_{10}$ & -&-&-&-&$c_0$&$x_0$ \\
\hline
$R_{11}$ & $x_{11}$ &-&-&-&-&$c_1$&$x_1$ \\
\hline
$R_{12}$ & $x_{12}$ &-&-&-&-&$c_2$&$x_2$ \\
\hline
\end{tabular}
\vspace{5pt}
\caption{Reduced complexity decoding for the SNC-SUICP given in Example \ref{ex5}}
\label{table2}
\vspace{-20pt}
\end{table*}
\end{example}

\begin{example}
\label{ex6}
Consider a one-sided SNC-SUICP with $K=44, D=17$. For this index coding problem, we have $\lambda_1=10,\lambda_2=7,\lambda_3=3,\lambda_4=1$ and $l=4$. The AIR matrix encoding matrix $\mathbf{L}_{44 \times 27}$ is given in Fig. \ref{ex4matrix}.

\begin{figure*}[t]
\arraycolsep=0.9pt
\begin{small}
{
$$\mathbf{L}_{44 \times 27}=\left[
\begin{array}{ccccccccccccccccccccccccccc}
1 & 0 & 0 & 0 & 0 & 0 & 0 & 0 & 0 & 0 & 0 & 0& 0 & 0 & 0 & 0 & 0 & 0 & 0 & 0 & 0 & 0 & 0 & 0 & 0 & 0 & 0\\
0 & 1 & 0 & 0 & 0 & 0 & 0 & 0 & 0 & 0 & 0 & 0& 0 & 0 & 0 & 0 & 0 & 0 & 0 & 0 & 0 & 0 & 0 & 0 & 0 & 0 & 0\\
0 & 0 & 1 & 0 & 0 & 0 & 0 & 0 & 0 & 0 & 0 & 0& 0 & 0 & 0 & 0 & 0 & 0 & 0 & 0 & 0 & 0 & 0 & 0 & 0 & 0 & 0\\
0 & 0 & 0 & 1 & 0 & 0 & 0 & 0 & 0 & 0 & 0 & 0& 0 & 0 & 0 & 0 & 0 & 0 & 0 & 0 & 0 & 0 & 0 & 0 & 0 & 0 & 0\\
0 & 0 & 0 & 0 & 1 & 0 & 0 & 0 & 0 & 0 & 0 & 0& 0 & 0 & 0 & 0 & 0 & 0 & 0 & 0 & 0 & 0 & 0 & 0 & 0 & 0 & 0\\
0 & 0 & 0 & 0 & 0 & 1 & 0 & 0 & 0 & 0 & 0 & 0& 0 & 0 & 0 & 0 & 0 & 0 & 0 & 0 & 0 & 0 & 0 & 0 & 0 & 0 & 0\\
 0 & 0 & 0 & 0 & 0  & 0 & 1 & 0 & 0 & 0 & 0& 0 & 0 & 0 & 0 & 0 & 0 & 0 & 0 & 0 & 0 & 0 & 0 & 0 & 0 & 0 & 0\\
 0 & 0 & 0 & 0 & 0 & 0 & 0 & {\color{red}\textbf{1}} & 0 & 0 & 0 & 0& 0 & 0 & 0 & 0 & 0 & 0 & 0 & 0 & 0 & 0 & 0 & 0 & 0 & 0 & 0\\
 0 & 0 & 0 & 0 & 0 & 0 & 0 & 0 & 1 & 0 & 0 & 0& 0 & 0 & 0 & 0 & 0 & 0 & 0 & 0 & 0 & 0 & 0 & 0 & 0 & 0 & 0\\
0 & 0 & 0 & 0 & 0 & 0 & 0 & 0 & 0 & 1 & 0 & 0& 0 & 0 & 0 & 0 & 0 & 0 & 0 & 0 & 0 & 0 & 0 & 0 & 0 & 0 & 0\\
 0 & 0 & 0 & 0 & 0 & 0 & 0 & 0 & 0 & 0 & 1 & 0& 0 & 0 & 0 & 0 & 0 & 0 & 0 & 0 & 0 & 0 & 0 & 0 & 0 & 0 & 0\\
 0 & 0 & 0 & 0 & 0 & 0 & 0 & 0 & 0 & 0 & 0 & 1& 0 & 0 & 0 & 0 & 0 & 0 & 0 & 0 & 0 & 0 & 0 & 0 & 0 & 0 & 0\\
 0 & 0 & 0 & 0 & 0 & 0 & 0 & 0 & 0 & 0 & 0 & 0& 1 & 0 & 0 & 0 & 0 & 0 & 0 & 0 & 0 & 0 & 0 & 0 & 0 & 0 & 0\\
0 & 0 & 0 & 0 & 0 & 0 & 0 & 0 & 0 & 0 & 0 & 0& 0 & 1 & 0 & 0 & 0 & 0 & 0 & 0 & 0 & 0 & 0 & 0 & 0 & 0 & 0\\
0 & 0 & 0 & 0 & 0 & 0 & 0 & 0 & 0 & 0 & 0 & 0& 0 & 0 & 1 & 0 & 0 & 0 & 0 & 0 & 0 & 0 & 0 & 0 & 0 & 0 & 0\\
 0 & 0 & 0 & 0 & 0 & 0 & 0 & 0 & 0 & 0 & 0 & 0& 0 & 0 & 0 & 1 & 0 & 0 & 0 & 0 & 0 & 0 & 0 & 0 & 0 & 0 & 0\\
0 & 0 & 0 & 0 & 0 & 0 & 0 & 0 & 0 & 0 & 0 & 0& 0 & 0 & 0 & 0 & 1 & 0 & 0 & 0 & 0 & 0 & 0 & 0 & 0 & 0 & 0 \\
 0 & 0 & 0 & 0 & 0 & 0 & 0 & 0 & 0 & 0 & 0 & 0& 0 & 0 & 0 & 0 & 0 & 1 & 0 & 0 & 0 & 0 & 0 & 0 & 0 & 0 & 0\\
 0 & 0 & 0 & 0 & 0 & 0 & 0 & 0 & 0 & 0 & 0 & 0& 0 & 0 & 0 & 0 & 0 & 0 &1 & 0 & 0 & 0 & 0 & 0 & 0 & 0 & 0\\
  0 & 0 & 0 & 0 & 0 & 0 & 0 & 0 & 0 & 0 & 0 & 0& 0 & 0 & 0 & 0 & 0 & 0 & 0 & 1 & 0 & 0 & 0 & 0 & 0 & 0 & 0\\
 0 & 0 & 0 & 0 & 0 & 0 & 0 & 0 & 0& 0  & 0 & 0 & 0& 0 & 0 & 0 & 0 & 0 & 0 & 0 & 1 & 0 & 0 & 0 & 0 & 0 & 0\\
 0 & 0 & 0 & 0 & 0 & 0 & 0 & 0 & 0 & 0 & 0 & 0 & 0& 0 & 0 & 0 & 0 & 0 & 0 & 0 & 0 & 1 & 0 & 0 & 0 & 0 & 0\\
  0 & 0 & 0 & 0 & 0 & 0 & 0 & 0 & 0 & 0 & 0 & 0& 0 & 0 & 0 & 0 & 0 & 0 & 0 & 0 & 0 & 0 & 1 & 0 & 0 & 0 & 0\\
 0 & 0 & 0 & 0 & 0 & 0 & 0 & 0 & 0 & 0 & 0 & 0& 0 & 0 & 0 & 0 & 0 & 0 & 0 & 0 & 0 & 0 & 0 & 1 & 0 & 0 & 0\\
 0 & 0 & 0 & 0 & 0 & 0 & 0 & 0 & 0 & 0 & 0 & 0& 0 & 0 & 0 & 0 & 0 & 0 & 0 & 0 & 0 & 0 & 0 & 0 & 1 & 0 & 0\\
  0 & 0 & 0 & 0 & 0 & 0 & 0 & 0 & 0 & 0 & 0 & 0& 0 & 0 & 0 & 0 & 0 & 0 & 0 & 0 & 0 & 0 & 0 & 0 & 0 & 1 & 0\\
  0 & 0 & 0 & 0 & 0 & 0 & 0 & 0 & 0 & 0 & 0 & 0& 0 & 0 & 0 & 0 & 0 & 0 & 0 & 0 & 0 & 0 & 0 & 0 & 0 & 0 & 1\\
  \hline
1 & 0 & 0 & 0 & 0 & 0 & 0 & 0 & 0 & 0 & 0 & 0& 0 & 0 & 0 & 0 & 0 & \vline 1 & 0 & 0 & 0 & 0 & 0 & 0 & 0 & 0 & 0\\
  0 & 1 & 0 & 0 & 0 & 0 & 0 & 0 & 0 & 0 & 0 & 0& 0 & 0 & 0 & 0 & 0 &\vline  0 & 1 & 0 &0 & 0 & 0 & 0 & 0 & 0 & 0\\
  0 & 0 & 1 & 0 & 0 & 0 & 0 & 0 & 0 & 0 & 0 & 0& 0 & 0 & 0 & 0 & 0 &\vline 0 & 0 &  1 & 0& 0 & 0 & 0 & 0 & 0 & 0\\
  0 & 0 & 0 & 1 & 0 & 0 & 0 & 0 & 0 & 0 & 0 & 0& 0 & 0 & 0 & 0 & 0 &\vline 0 & 0 & 0 &  1 & 0 & 0& 0 & 0& 0 & 0\\
 0 & 0 & 0 & 0 & 1 & 0 & 0 & 0 & 0 & 0 & 0 & 0& 0 & 0 & 0 & 0 & 0 &\vline 0 & 0 &0 & 0 & 1& 0 & 0 & 0 & 0 & 0\\
  0 & 0 & 0 & 0 & 0 & 1 & 0 & 0 & 0 & 0 & 0 & 0& 0 & 0 & 0 & 0 & 0 &\vline 0 & 0 & 0 & 0 & 0 & 1 & 0 & 0 & 0 & 0\\
  0 & 0 & 0 & 0 & 0  & 0 & 1 & 0 & 0 & 0 & 0& 0 & 0 & 0 & 0 & 0 & 0 &\vline 0 & 0 & 0 & 0 & 0 & 0 & 1 & 0 & 0 & 0 \\
 0 & 0 & 0 & 0 & 0 & 0 & 0 & {\color{red}\textbf{1}} & 0 & 0 & 0 & 0& 0 & 0 & 0 & 0 & 0 &\vline 0 & 0 & 0 & 0 & 0 & 0 & 0 & {\color{red}\textbf{1}} & 0 & 0\\
 0 & 0 & 0 & 0 & 0 & 0 & 0 & 0 & 1 & 0 & 0 & 0& 0 & 0 & 0 & 0 & 0 &\vline 0 & 0 & 0 & 0 & 0 & 0 & 0 & 0 & 1 & 0\\
 0 & 0 & 0 & 0 & 0 & 0 & 0 & 0 & 0 & 1 & 0 & 0& 0 & 0 & 0 & 0 & 0 &\vline  \underline{0} & \underline{0} & \underline{0} & \underline{0} & \underline{0} & \underline{0} & \underline{0} & \underline{0} & \underline{0} & \underline{1}\\ 
 0 & 0 & 0 & 0 & 0 & 0 & 0 & 0 & 0 & 0 & 1 & 0& 0 & 0 & 0 & 0 & 0 &\vline 1 & 0 & 0 & 0 & 0 & 0 & 0 &\vline {\color{red}\textbf{1}} & 0 & 0\\
 0 & 0 & 0 & 0 & 0 & 0 & 0 & 0 & 0 & 0 & 0 & 1& 0 & 0 & 0 & 0 & 0 &\vline 0 &1 & 0 & 0 & 0 & 0 & 0 &\vline 0 & 1 & 0\\
 0 & 0 & 0 & 0 & 0 & 0 & 0 & 0 & 0 & 0 & 0 & 0& 1 & 0 & 0 & 0 & 0 &\vline 0 & 0 & 1 & 0  & 0 & 0 & 0 &\vline 0 & 0 & 1\\
 0 & 0 & 0 & 0 & 0 & 0 & 0 & 0 & 0 & 0 & 0 & 0& 0 & 1 & 0 & 0 & 0 &\vline 0 & 0 & 0 & 1 & 0 & 0 & 0 &\vline {\color{red}\textbf{1}} & 0 & 0\\
0 & 0 & 0 & 0 & 0 & 0 & 0 & 0 & 0 & 0 & 0 & 0& 0 & 0 & 1 & 0 & 0 &\vline 0 & 0 & 0 & 0 & 1 & 0 & 0 &\vline 0 & 1 & 0\\
0 & 0 & 0 & 0 & 0 & 0 & 0 & 0 & 0 & 0 & 0 & 0& 0 & 0 & 0 & 1 & 0 &\vline 0 & 0 & 0 & 0 & 0 & 1 & 0 &\vline \underline{0} & \underline{0} & \underline{1}\\
0 & 0 & 0 & 0 & 0 & 0 & 0 & 0 & 0 & 0 & 0 & 0& 0 & 0 & 0 & 0 & 1 &\vline 0 & 0 & 0 & 0 & 0 & 0 & 1 &\vline {\color{red}\textbf{1}} & 1 & 1\\
 \end{array}
\right]$$
\caption{Encoding matrix for the SNC-SUICP in Example \ref{ex4}.}
\label{ex4matrix}
}
\end{small}
\end{figure*}
Let $k=7$. Receiver $R_7$ wants to decode $x_7$. We have, $d_{down}(7)=27, \mu_7=d_{right}(7+27,7)=d_{right}(34,7)=17,t_{7,1}=3,t_{7,2}=6,t_{7,3}=9$. The receiver $R_7$ decodes $x_7$ by adding the broadcast symbols $c_7,c_{10}, c_{13}, c_{16}$ and $c_{24}$. From the AIR matrix  we have 
\begin{align*}
&c_7=x_7+x_{34} \\&
c_{10}=x_{10}+x_{37}\\&
c_{13}=x_{13}+x_{40}\\&
c_{16}=x_{16}+x_{43}\\&
c_{24}=x_{24}+x_{34}+x_{37}+x_{40}+x_{43}\\&
c_7+c_{10}+c_{13}+c_{16}+c_{24}=x_7+x_{10}+x_{13}+x_{16}+x_{24}.
\end{align*}
 
In $c_7+c_{10}+c_{13}+c_{16}+c_{24}$, the message symbols $x_{10},x_{13},x_{16}$ and $x_{24}$ are in the side-information of $R_7$. Thus, $R_7$ decodes its wanted message $x_7$. 

Table \ref{table3} gives the broadcast symbols and side-information $\gamma_k$ used by receiver $R_k$ to decode $x_k$ for all $k\in [0:43]$.
\begin{table*}
\centering
\setlength\extrarowheight{0.0pt}
\begin{tabular}{|c|c|c|c|c|c|c|c|}
\hline
\textbf{$R_k$} &$\mathcal{W}_k$&$d_{down}$&$\mu_k$&$p_k$& $t_{k,\tau}$&$c_k,c_{k+t_{k,\tau}},c_{k+\mu_k}$&$\gamma_k$ \\
\cline{6-7}
&&(k)&&& \multicolumn{1}{c}{}&$\tau=1,2,\ldots,p_k$&~\\
\hline
$R_0$ & $x_0$ & 27&17&1&$10$&$c_0,c_{10},c_{17}$&$x_{10},x_{17}$ \\
\hline
$R_1$ & $x_1$ & 27&17&1&$10$&$c_1,c_{11},c_{18}$&$x_{11},x_{18}$ \\
\hline
$R_2$ & $x_2$ & 27&17&1&$10$&$c_2,c_{12},c_{19}$&$x_{12},x_{19}$ \\
\hline
$R_3$ & $x_3$ & 27&17&1&$10$&$c_3,c_{13},c_{20}$&$x_{13},x_{20}$ \\
\hline
$R_4$ & $x_4$ & 27&17&1&$10$&$c_4,c_{14},c_{21}$&$x_{14},x_{21}$ \\
\hline
$R_5$ & $x_5$ & 27&17&1&$10$&$c_5,c_{15},c_{22}$&$x_{15},x_{22}$ \\
\hline
$R_6$ & $x_6$ & 27&17&1&$10$&$c_6,c_{16},c_{23}$&$x_{16},x_{23}$ \\
\hline
$R_7$ & $x_7$ & 27&17&3&$3,6,9$&$c_7,c_{10},c_{13},c_{16},c_{24}$&$x_{10},x_{13},x_{16},x_{24}$ \\
\hline
$R_8$ & $x_8$ &27&17&3&$3,6,8$&$c_8,c_{11},c_{14},c_{16},c_{25}$&$x_{11},x_{14},x_{16},x_{25}$ \\
\hline
$R_{9}$ & $x_{9}$ & 27&17&3&$3,6,7$&$c_9,c_{12},c_{15},c_{16},c_{26}$&$x_{12},x_{15},x_{16},x_{26}$ \\
\hline
$R_{10}$ & $x_{10}$ & 27&7&0&$0$&$c_{10},c_{17}$&$x_{17},x_{27}$ \\
\hline
$R_{11}$ & $x_{11}$ & 27&7&0&$0$&$c_{11},c_{18}$&$x_{18},x_{18}$ \\
\hline
$R_{12}$ & $x_{12}$ & 27&7&0&$0$&$c_{12},c_{19}$&$x_{19},x_{29}$ \\
\hline
$R_{13}$ & $x_{13}$ & 27&7&0&$0$&$c_{13},c_{20}$&$x_{20},x_{30}$ \\
\hline
$R_{14}$ & $x_{14}$ & 27&7&0&$0$&$c_{14},c_{21}$&$x_{21},x_{31}$ \\
\hline
$R_{15}$ & $x_{15}$ & 27&7&0&$0$& $c_{15},c_{22}$&$x_{22},x_{32}$\\
\hline
$R_{16}$ & $x_{16}$ & 27&7&0&$0$& $c_{16},c_{23}$&$x_{23},x_{33}$\\
\hline
$R_{17}$ & $x_{17}$ & 20&7&2&$3,6$& $c_{17},c_{20},c_{23},c_{24}$&$x_{20},x_{23},x_{24},x_{27},x_{30},x_{33},x_{34}$\\
\hline
$R_{18}$ & $x_{18}$ & 20&7&2&$3,5$&$c_{18},c_{21},c_{23},c_{25}$&$x_{21},x_{23},x_{25},x_{28},x_{31},x_{33},x_{35}$ \\
\hline
$R_{19}$ & $x_{19}$ & 20&7&2&$3,4$&$c_{19},c_{22},c_{23},c_{26}$&$x_{22},x_{23},x_{26},x_{29},x_{32},x_{33},x_{36}$ \\
\hline
$R_{20}$ & $x_{20}$ & 20&4&1&$3$& $c_{20},c_{23},c_{24}$&$x_{23},x_{24},x_{30},x_{33},x_{34},x_{37}$\\
\hline
$R_{21}$ & $x_{21}$ &20&4&1&$2$& $c_{21},c_{23},c_{25}$&$x_{23},x_{25},x_{31},x_{33},x_{35},x_{38}$\\
\hline
$R_{22}$ & $x_{22}$ &20&4&1&$1$& $c_{22},c_{23},c_{26}$&$x_{23},x_{26},x_{32},x_{33},x_{36},x_{39}$\\
\hline
$R_{23}$ & $x_{23}$ &20&1&0&$0$&$c_{23},c_{24}$&$x_{24},x_{33},x_{34},x_{37},x_{40}$ \\
\hline
$R_{24}$ & $x_{24}$ &19&1&0&$0$& $c_{24},c_{25}$&$x_{25},x_{34},x_{35},x_{37},x_{38},x_{40},x_{41}$\\
\hline
$R_{25}$ & $x_{25}$ &18&1&0&$0$&$c_{25},c_{26}$&$x_{26},x_{35},x_{36},x_{38},x_{39},x_{41},x_{42}$ \\
\hline
$R_{26}$ & $x_{26}$ &17&-&0&$0$& $c_{26}$&$x_{36},x_{39},x_{42},x_{43}$\\
\hline
$R_{27}$ & $x_{27}$ &-&-&-&-& $c_{0}$&$x_{0}$\\
\hline
$R_{28}$ & $x_{28}$ &-&-&-&-& $c_{1}$&$x_{1}$\\
\hline
$R_{29}$ & $x_{29}$ &-&-&-&-& $c_{2}$&$x_{2}$\\
\hline
$R_{30}$ & $x_{30}$ &-&-&-&-& $c_{3}$&$x_{3}$\\
\hline
$R_{31}$ & $x_{31}$ &-&-&-&-& $c_{4}$&$x_{4}$\\
\hline
$R_{32}$ & $x_{32}$ &-&-&-&-& $c_{5}$&$x_{5}$\\
\hline
$R_{33}$ & $x_{33}$ &-&-&-&-& $c_{6}$&$x_{6}$\\
\hline
$R_{34}$ & $x_{34}$ &-&-&-&-& $c_{7}$&$x_{7}$\\
\hline
$R_{35}$ & $x_{35}$ &-&-&-&-& $c_{8}$&$x_{8}$\\
\hline
$R_{36}$ & $x_{36}$ &-&-&-&-& $c_{9}$&$x_{9}$\\
\hline
$R_{37}$ & $x_{37}$ &-&-&-&-& $c_{10}$&$x_{10}$\\
\hline
$R_{38}$ & $x_{38}$ &-&-&-&-& $c_{11}$&$x_{11}$\\
\hline
$R_{39}$ & $x_{39}$ &-&-&-&-& $c_{12}$&$x_{12}$\\
\hline
$R_{40}$ & $x_{40}$ &-&-&-&-& $c_{13}$&$x_{13}$\\
\hline
$R_{41}$ & $x_{41}$ &-&-&-&-& $c_{14}$&$x_{14}$\\
\hline
$R_{42}$ & $x_{42}$ &-&-&-&-& $c_{15}$&$x_{15}$\\
\hline
$R_{43}$ & $x_{43}$ &-&-&-&-& $c_{16}$&$x_{16}$\\
\hline
\end{tabular}
\vspace{5pt}
\caption{Reduced complexity decoding for the SNC-SUICP given in Example \ref{ex4}}
\label{table3}
\vspace{0pt}
\end{table*}
\end{example}
\begin{example}
\label{ex7}
Consider the one-sided SNC-SUICP with $K=17, D=7$. For this index coding problem, $l=2, \lambda_1=3, \lambda_2=1$. The AIR matrix of size $17 \times 10$ was seen in Section \ref{sec2}. For this SNC-SUICP, Table \ref{table4} gives the broadcast symbols and side-information $\gamma_k$ used by receiver $R_k$ to decode $x_k$ for all $k\in [0:16]$. 
\end{example}
\begin{table*}
\centering
\setlength\extrarowheight{2pt}
\begin{tabular}{|c|c|c|c|c|c|c|c|}
\hline
\textbf{$R_k$} &$\mathcal{W}_k$&$d_{down}(k)$&$\mu_k$&$p_k$& $t_{k,\tau}$&$c_k,c_{k+t_{k,\tau}},c_{k+\mu_k}$&$\gamma_k$\\
\cline{6-7}
&&&&& \multicolumn{1}{c}{}&$\tau=1,2,\ldots,p_k$&~\\
\hline
$R_0$ & $x_0$ & 10&7&2&3,6&$c_0,c_3,c_6,c_{7}$&$x_3,x_6,x_7$ \\
\hline
$R_1$ & $x_1$ & 10&7&2&3,5&$c_1,c_4,c_6,c_8$&$x_3,x_6,x_7$\\
\hline
$R_2$ & $x_2$ & 10&7&2&3,4&$c_2,c_5,c_6,c_9$&$x_5,x_6,x_{9}$\\
\hline
$R_3$ & $x_3$ & 10&4&1&3&$c_3,c_6,c_7$&$x_6,x_7,x_{10}$ \\
\hline
$R_4$ & $x_4$ & 10&4&1&2&$c_4,c_6,c_8$&$x_6,x_8,x_{11}$ \\
\hline
$R_5$ & $x_5$ & 10&4&1&1&$c_5,c_6,c_9$&$x_6,x_{9},x_{12}$ \\
\hline
$R_6$ & $x_6$ & 10&1&0&0&$c_6,c_7$& $x_{7},x_{10},x_{13}$ \\
\hline
$R_7$ & $x_7$ & 9&1&0&0&$c_7,c_8$& $x_{9},x_{11},x_{12},x_{14},x_{15}$\\
\hline
$R_8$ & $x_8$ &8&1&0&0&$c_8,c_9$& $x_{12},x_{15},x_{16}$ \\
\hline
$R_{9}$ & $x_{9}$ &7&-&0&0&$c_9$& $x_0$\\
\hline
$R_{10}$ & $x_{10}$ &-&-&-&-&$c_0$& $x_0$\\
\hline
$R_{11}$ & $x_{11}$ &-&-&-&-&$c_1$& $x_1$\\
\hline
$R_{12}$ & $x_{12}$ &-&-&-&-&$c_2$& $x_2$\\
\hline
$R_{13}$ & $x_{13}$ &-&-&-&-&$c_3$& $x_3$\\
\hline
$R_{14}$ & $x_{14}$ &-&-&-&-&$c_4$& $x_4$\\
\hline
$R_{15}$ & $x_{15}$ &-&-&-&-&$c_5$& $x_5$\\
\hline
$R_{16}$ & $x_{16}$ &-&-&-&-&$c_6$& $x_6$\\
\hline
\end{tabular}
\vspace{5pt}
\caption{Reduced complexity decoding for the SNC-SUICP given in Example \ref{ex7}}
\label{table4}
\vspace{-20pt}
\end{table*}
\section{Application to Noisy Index Coding}
\label{sec5}

In a noisy index coding problem with binary transmission, index coded bits are transmitted over $K$  binary symmetric channels, ( $K$ is the number of users) each with independent and identical cross-over probability $p$.The value of $p$ is determined by the channel noise statistics,like AWGN, Rayleigh fading, Rician fading etc. 

 Let $\mathfrak{C}=\{ c_{0},c_{1},\ldots, c_{N-1} \}$ be the $N-1$ broadcast symbols ($c_i \in \mathbb{F}_2$ for $i\in [0:N-1]$), which are required to be broadcasted to $K$ receivers. Let $\{ s_{0},s_{1},\ldots,s_{N-1} \}$, denote the sequence of binary antipodal channel symbols transmitted by the source.

For instance, considering  a quasi-static fading broadcast channel, the received symbol sequence  at receiver $R_{k}$ corresponding to the transmission of $s_{i}$ is given by $y_{k}=h_{k}s_{i}+w_k$ where $h_{k}$ is the fading coefficient associated with the channel from source to receiver $R_{k}$ and $w_k$ is the sample of additive white Gaussian noise process. Receiver $R_k$ decodes its wanted message symbol $x_k$ in the following two steps.
\begin{itemize}
\item $R_k$ first estimates $N$ broadcast symbols from $y_0,y_1,\ldots,y_{N-1}$. Let $\tilde{c}_0,\tilde{c}_1,\ldots,\tilde{c}_{N-1}$ be the estimated broadcast symbols, where $\tilde{c}_i \in \mathbb{F}_2$. A broadcast symbol error occurs with probability $p$.
\item By using $\tilde{c}_0,\tilde{c}_1,\ldots,\tilde{c}_{N-1}$ and the side-information available with it, receiver $R_k$ decodes its wanted message. Let $\tilde{x}_k$ be the message symbol decoded by the receiver $R_k$. The message is in error if $\tilde{x}_k \neq x_k.$ 
\end{itemize}

In \cite{TRCR}, it is shown that the message probability of error in decoding a message at a particular receiver decreases with a decrease in the number of transmissions used to decode the message among the total of broadcast transmissions. Also an algorithm to identify an optimal index code for single uniprior index coding problems which gives the best performance in terms of minimal maximum error probability across all the receivers is given. 




For one-sided SNC-SUICP, Theorem \ref{thm2} explicitly gives the number of broadcast symbols used by each receiver to decode its wanted message.  The number of broadcast symbols used by each receiver to decode its wanted message is summarized below. 
\begin{itemize}
\item If $k \in [K-D-\lambda_l:K-1]$, then the number of broadcast symbols required by $R_k$ to decode $x_k$ is one.
\item If $k \in [0:K-D-\lambda_l-1]$, then the number of broadcast symbols required by $R_k$ to decode $x_k$ is $p_k+2$, where $p_k$ is the number of $1$s below $\mathbf{L}(k+d_{down}(k),k+\mu_k)$ in the AIR matrix.
\end{itemize}    
For the codes discussed in Examples \ref{ex5}, \ref{ex6} and \ref{ex7},  Tables \ref{table2} \ref{table3} and \ref{table4} show  the exact number of transmissions used by each receiver to decode its wanted message in the 7th column. In the following subsection we elaborate this aspect with simulation results for Example \ref{ex5}. 

Another application of our results in noise index coding problem is related to Instantly Decodable Network Coding (IDNC). IDNC deals with code designs when the receivers have no buffer and need to decode the wanted messages instantly without having stored previous transmissions. A recent survey article on IDNC with application to Device-to-Device (D2D) communications is \cite{DSAA}. These results are valid for index coding since it is a special case of network coding. The decoding expressions we have developed in this paper identifies the receivers which can decode their wanted messages using only one broadcasted transmission. In other words we have shown that the receivers $R_k$ with  $k \in [K-D-\lambda_l:K-1]$ can decode their messages without using any buffer. These receivers can be seen in the 7th column of Tables \ref{table2}, \ref{table3} and \ref{table4} for Examples \ref{ex5}, \ref{ex6} and \ref{ex7} respectively.  

\subsection{Simulation Results}
In this subsection, we give simulation results which show that the number of broadcast symbols used by receivers to decode its wanted message matters. We show that the receiver which uses lesser number of broadcast symbols to decode its wanted message perform better than those using more number of broadcast symbols. 

We consider the SNC-SUICP given in Example \ref{ex5} to simulate the SNR vs BER performance of each receiver. In the simulation, the transmitter uses BPSK signal set. The performance of receivers were evaluated in the AWGN scenario and Rayleigh fading scenario with perfect channel state information at receiver. The performance of receiver $R_k$, $k \in [0:12]$ in Example \ref{ex5} for AWGN channel is shown in Fig. \ref{afig114}. The performance of receiver $R_k$, $k \in [0:12]$ in Example \ref{ex5} for Rayleigh fading channel is shown in  Fig. \ref{afig113}.

In Fig. \ref{afig114} and \ref{afig113}, we observe that at any given SNR, maximum error probability occurs at receiver $R_6$ and minimum error probability occurs at receivers $R_{9},R_{10},R_{11}$ and $R_{12}$. We can also observe that the SNR vs BER plots of receivers $R_0,R_1,R_2,R_3,R_4,R_5$ and $R_8$ are overlapping with each other. From Table \ref{table2}, we see that $R_6$ is using four broadcast symbols to decode its wanted message and receivers $R_{9},R_{10},R_{11},R_{12}$ using only one broadcast symbol to decode their wanted message symbol. We also observe that the receivers $R_0,R_1,R_2,R_3,R_4,R_5$ and $R_8$ are using equal number of broadcast symbols to decode their wanted message symbol. Hence, from simulations, it is evident that the receiver which use lesser number of broadcast symbols to decode its wanted message perform better than those using more number of broadcast symbols.

The low-complexity decoding given in Theorem \ref{thm2} explicitly quantifies the number of broadcast symbols and the number of side-information required by each receiver to decode its wanted message. Hence, Theorem \ref{thm2} gives the receiver with best error performance and the receiver with worst error performance. In fact, based on number of broadcast symbols used by a receiver to decode its wanted message, Theorem \ref{thm2} can partition the $K$ receivers into disjoint sets such that the error performance in every receiver in a set is equal. In Fig. \ref{afig114} and \ref{afig113}, we can observe that the error performance of every receiver in the four sets $\{R_{9},R_{10},R_{11},R_{12}\}$, $\{R_0,R_1,R_2,R_3,R_4,R_5,R_8\}$, $\{R_6\}$ and $\{R_7\}$ is equal and hence we can notice four curves in the simulation results.

In Rayleigh fading channel, simulation results show that the receiver $R_6$ requires 5.3 dB more SNR than the receivers $R_{9},R_{10},R_{11},R_{12}$ to get same probability of error. The receivers $R_0,R_1,R_2,R_3,R_4,R_5,R_8$ require 2.7 dB more SNR  than the receivers $R_{9},R_{10},R_{11},R_{12}$ to get same probability of error.
In AWGN channel, simulation results show that the receiver $R_6$ requires 1.0 dB more SNR than the receivers $R_{9},R_{10},R_{11},R_{12}$ to get a  BER of $2 \times 10^{-4}$. The receivers $R_0,R_1,R_2,R_3,R_4,R_5,R_8$ require 0.5 dB more SNR  than the receivers $R_{9},R_{10},R_{11},R_{12}$ to get a BER of $2 \times 10^{-4}$.
\begin{figure}
\centering
\hspace*{-2.0cm}
\includegraphics[scale=0.60]{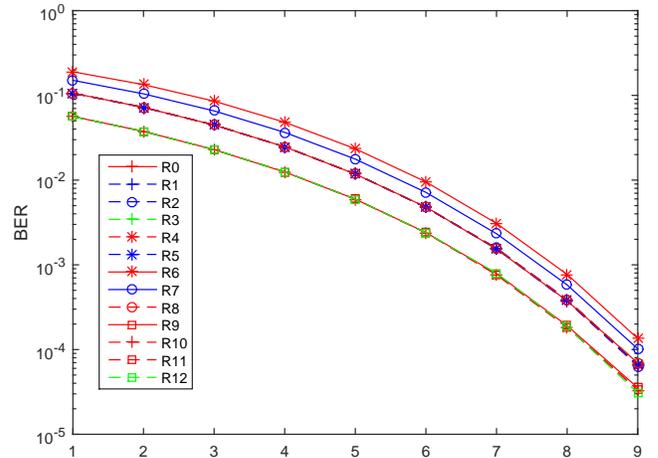}
\caption{SNR Vs BER in AWGN channel at receivers $R_k$ for $k \in [0:12]$ for Example \ref{ex5}.}
\label{afig114}
\end{figure}

\begin{figure}
\centering
\includegraphics[scale=0.6]{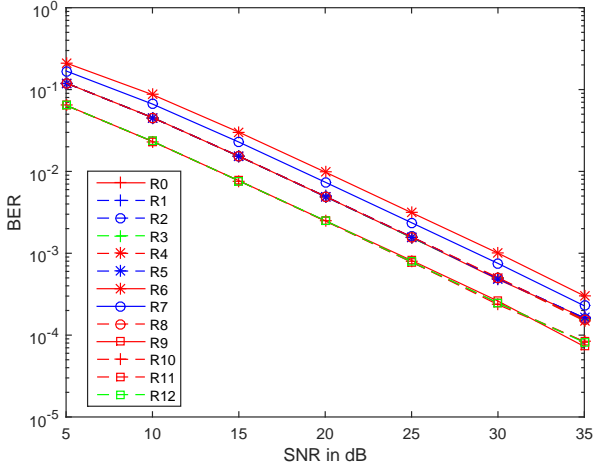}
\caption{SNR Vs BER in Rayleigh fading channel at receivers $R_k$ for $k \in [0:12]$ for Example \ref{ex5}.}
\label{afig113}
\end{figure}
\section{Discussion}
\label{sec6}
In this paper, we have presented a low-complexity decoding for one-sided SNC-SUICP. Our decoding procedure explicitly identified the side-information subset as well as the subset of broadcast transmissions to be used by any receiver. Since the number of broadcast transmissions to be used by every receiver is known apriori, when the index code is used in a noisy broadcast channel, its relative probability of error performance among different receivers is predictable. Further more the set of receivers which would use exactly one broadcast transmission was identified.  Some of the interesting directions of further research are as follows:  
\begin{itemize}
\item 
It will be interesting to characterize the class of non-symmetric index codes described in Corollary \ref{cor1} and the their side-information graphs. We conjecture that these non-symmetric codes are  critical index coding problems \cite{TSG}. 
\item The parameters $l$ as well as the down-distance parameter $p_k$ play important roles in getting the Boolean expression as well as number of broadcast transmissions used by a receiver. Closed form expressions or bounds for these for a given set of values of $K$ and $D$ will be very useful.
\item A natural direction for further research is to extend the results of this paper to two-sided SNC-SUICP and vector linear codes.
\end{itemize}
\section*{APPENDIX A}
\subsection*{Proof of Theorem \ref{thm1}}

\textit{Case (i)}: Let $k \in C_i$ for $i \in [0:\lceil\frac{l}{2}\rceil]$ if $l$ is even and  $i \in [0:\lceil\frac{l}{2}\rceil-1]$ if  $l$ is odd. 

In this case, as shown in Figure \ref{bfig1}, the $k$th column of $\mathbf{L}$ intersects with $\mathbf{I}_{K-D},\mathbf{I}_{\lambda_{2i} \times  \beta_{2i}\lambda_{2i}}$ and $\mathbf{I}_{\beta_{2j-1}\lambda_{2j-1} \times \lambda_{2j-1}}$ for $j=1,2,\ldots,i$. As can be seen from Figure \ref{bfig1}, the $k$th column comprises of $2+\beta_1\lambda_1+\beta_3\lambda_3+\ldots+\beta_{2i-1}\lambda_{2i-1}$ number of $1$s, the positions of which are given below.
\begin{itemize}
\item $k$
\item $k+j\lambda_1$ for $j=1,2,\ldots,\beta_1$
\item $k+\beta_1\lambda_1+j\lambda_3$ for $j=1,2,\ldots,\beta_3$
\item $k+\beta_1\lambda_1+\beta_3\lambda_3+j\lambda_5$ for $j=1,2,\ldots,\beta_5$ \\
~~~~~~~~~~~~~~~~~~~~~$\ldots$
\item $k+\beta_1\lambda_1+\beta_3\lambda_3+\ldots+\beta_{2i-3}\lambda_{2i-3}+j\lambda_{2i-1}$ for $j=1,2,\ldots,\beta_{2i-1}$
\item $(K-\lambda_{2i})+(k-(K-D-\lambda_{2i-1}))~\text{\textit{mod}}~\lambda_{2i}.$
\end{itemize}

Since, the $k$th broadcast symbol is obtained by multiplying $k$th column of $\mathbf{L}$ with the message vector $[x_0~x_1~\ldots~x_{K-1}]$, the $k$th broadcast symbol can be written as 
\begin{align}
\label{apexcode1}
\nonumber
c_k=x_k&+\sum_{j=1}^{\beta_1} x_{k+j\lambda_1}+\sum_{j=1}^{\beta_3} x_{k+\beta_1\lambda_1+j\lambda_3}\\&
\nonumber
+\sum_{j=1}^{\beta_5} x_{k+\beta_1\lambda_1+\beta_3\lambda_3+j\lambda_5}\\&
\nonumber
+\ldots \\&
\nonumber
+\sum_{j=1}^{\beta_{2i-1}} x_{k+\beta_1\lambda_1+\beta_3\lambda_3+\ldots+\beta_{2i-3}\lambda_{2i-3}+j\lambda_{2i-1}}\\&
+x_{(K-\lambda_{2i})+(k-K+D+\lambda_{2i-1})~\text{\textit{mod}}~\lambda_{2i} }.
\end{align}
\textit{Case (ii)}: $l$ is odd and $k \in C_{\lceil\frac{l}{2}\rceil}$

In this case, the $k$th column of $\mathbf{L}$ intersects with $\mathbf{I}_{K-D}$ and $\mathbf{I}_{\beta_{2i-1}\lambda_{2i-1} \times \lambda_{2i-1}}$ for $i=1,2,\ldots,2\lceil\frac{l}{2}\rceil$. The $k$th column comprises of $1+\beta_1\lambda_1+\beta_3\lambda_3+\ldots+\beta_{l}\lambda_{l}$ number of $1$s. The positions of $1$s in $k$th column are given below.
\begin{itemize}
\item $k$
\item $k+j\lambda_1$ for $j=1,2,\ldots,\beta_1$
\item $k+\beta_1\lambda_1+j\lambda_3$ for $j=1,2,\ldots,\beta_3$
\item $k+\beta_1\lambda_1+\beta_3\lambda_3+j\lambda_5$ for $j=1,2,\ldots,\beta_5$\\
~~~~~~~~~~~ $\ldots$
\item $k+\beta_1\lambda_1+\beta_3\lambda_3+\ldots+j\lambda_{l}$ for $j=1,2,\ldots,\beta_{l}$.
\end{itemize}
Hence, the $k$th broadcast symbol can be written as 
\begin{align}
\label{apexcode2}
\nonumber
c_k=x_k&+\sum_{j=1}^{\beta_1} x_{k+j\lambda_1}+\sum_{j=1}^{\beta_3} x_{k+\beta_1\lambda_1+j\lambda_3}\\&
\nonumber
+\sum_{j=1}^{\beta_5} x_{k+\beta_1\lambda_1+\beta_3\lambda_3+j\lambda_5}+\ldots \\&
+\sum_{j=1}^{\beta_{l}} x_{k+\beta_1\lambda_1+\ldots+\beta_{l-2}\lambda_{l-2}+j\lambda_{l}}.
\end{align}

The number of terms in \eqref{apexcode1} depends on the value of $k$. That is, the term $\sum_{j=1}^{\beta_1} x_{k+j\lambda_1}$ is present in $c_k$ if $k \in C_1 \cup C_2 \cup \ldots \cup C_{\lceil\frac{l}{2}\rceil}$. The term $\sum_{j=1}^{\beta_3} x_{k+\beta_1\lambda_1+j\lambda_3}$ is present in $c_k$ if $k \in C_2\cup C_3 \cup \ldots \cup C_{\lceil\frac{l}{2}\rceil}$. Indicator function can be used to mark the presence of a given term in $c_k$. By combining \eqref{apexcode1} and \eqref{apexcode2}, we get the Boolean expression given in \eqref{code}.
\begin{figure*}
\centering
\includegraphics[scale=0.6]{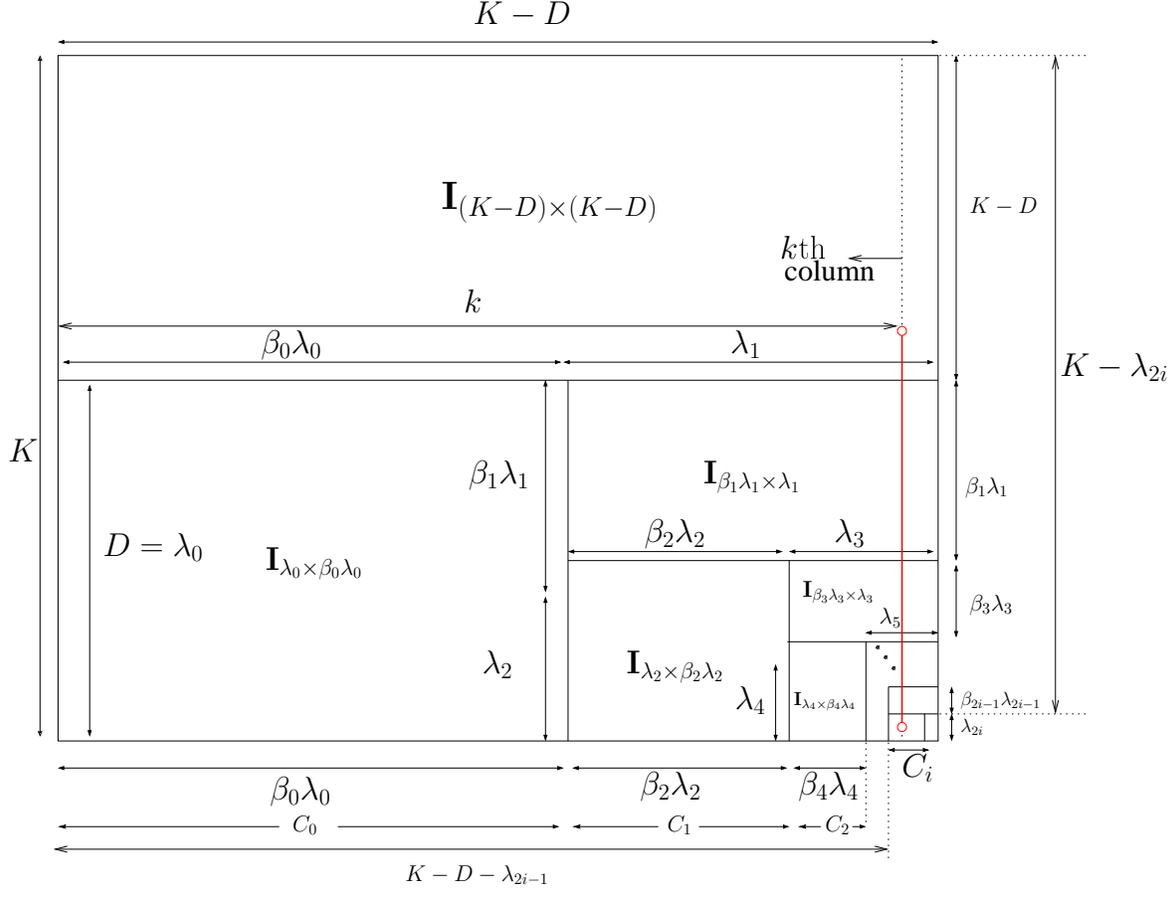}\\
\caption{Boolean Expression for $k$th broadcast symbol}
\label{bfig1}
\end{figure*}
\section*{APPENDIX B}
\subsection*{Proof of Lemma \ref{lemma1} }
\textit{Case (i)}: $l$ is even and $k \in C_i$ for $i \in [0:\lceil\frac{l}{2}\rceil]$ or  $l$ is odd and $k \in C_i$ for $i \in [0:\lceil\frac{l}{2}\rceil-1]$.

In this case, from the definition of down distance, we have $L(k+d_{down}(k),k) \in \mathbf{I}_{ \lambda_{2i} \times \beta_{2i}\lambda_{2i}}$.
\begin{figure*}
\centering
\includegraphics[scale=0.67]{matrix.eps}\\
\caption{Maximum-down distance calculation}
\label{afig1}
\end{figure*}
Let $k~\text{mod}~(K-D-\lambda_{2i-1})=c\lambda_{2i}+d$ for some positive integers $c$ and $d$ $(d<\lambda_{2i})$. 
From Figure \ref{afig1}, we have 
\begin{align}
\label{aeq1}
d_{down}(k)=d_1+d_2+d_3,
\end{align}
and 
\begin{align}
\label{aeq2}
\nonumber
&d_1=(K-D)-k, \\&
\nonumber
d_2=D-\lambda_{2i}, \\&
d_3=k-(K-D-\lambda_{2i-1})-c\lambda_{2i}.
\end{align}

By using \eqref{aeq1} and \eqref{aeq2}, we have 
\begin{align}
\label{aeq3}
\nonumber
d_{down}(k)&=d_1+d_2+d_3\\&
\nonumber
=(K-D)-k+\underbrace{D-\lambda_{2i}}_{d_2}\\&+\underbrace{k-(K-D-\lambda_{2i-1})-c\lambda_{2i}}_{d_3}\\&
\nonumber
=D+\lambda_{2i-1}-(c+1)\lambda_{2i}.
\end{align}

By replacing $\lambda_{2i-1}$ with $\beta_{2i}\lambda_{2i}+\lambda_{2i+1}$ in \eqref{aeq3}, we get 
\begin{align*}
d_{down}(k)=D+\lambda_{2i+1}+(\beta_{2i}-1-c)\lambda_{2i}.
\end{align*}
\textit{Case (ii)}: $l$ is odd and $k \in C_{\lceil\frac{l}{2}\rceil}$.

In this case, from the definition of down distance, we have $L(k+d_{down}(k),k) \in \mathbf{I}_{\beta_{l}\lambda_{l} \times \lambda_{l}}$.
\begin{figure*}
\centering
\includegraphics[scale=0.67]{amatrix.eps}\\
\caption{Maximum-down distance calculation}
\label{afig2}
\end{figure*}
From Figure \ref{afig2}, we have 
\begin{align}
\label{aeq32}
d_{down}(k)=d_1+d_2+d_3,
\end{align}
and
\begin{align}
\label{aeq4}
\nonumber
&d_1=(K-D)-k, \\&
\nonumber
d_2=D-\beta_l\lambda_{l}, \\&
d_3=\beta_l\lambda_l-d_5.
\end{align}

We have $L(k,k) \in \mathbf{I}_{K-D}$ and $L(k+d_{down}(k),k) \in \mathbf{I}_{\lambda_l}$ of $\mathbf{I}_{\beta_{l}\lambda_{l} \times \lambda_{l}}$ as shown in Figure \ref{afig2}. Hence, we have $d_1=d_4$ and $d_4=d_5$.
By using \eqref{aeq32} and \eqref{aeq4}, we have 
\begin{align*}
\nonumber
d_{down}(k)&=d_1+d_2+d_3\\&
\nonumber
=d_1+\underbrace{D-\beta_l\lambda_{l}}_{d_2}+\underbrace{\beta_l\lambda_l-d_1}_{d_3}=D.
\end{align*}

For $i=\left \lceil \frac{l}{2} \right \rceil$, we have $\lambda_{2\left \lceil \frac{i}{2} \right \rceil}=\lambda_{2\left \lceil \frac{i}{2} \right \rceil+1}=0$. We can write $D$ as $D+\lambda_{2\left \lceil \frac{l}{2} \right \rceil+1}+(\beta_{2\left\lceil\frac{l}{2}\right\rceil}-1-c)\lambda_{2\left \lceil \frac{l}{2} \right \rceil}$. Hence 
\begin{align*}
d_{down}(k)=D+\lambda_{2i+1}+(\beta_{2i}-1-c)\lambda_{2i}.
\end{align*}

\section*{APPENDIX C}
\subsection*{Proof of Lemma \ref{lemma2}}

\textit{Case (i)}: $\mathbf{L}(j,k) \in \mathbf{I}_{\beta_{2i+1} \lambda_{2i+1} \times \lambda_{2i+1}}$  (one of the odd-submatrices) for $i=0,1,\ldots,\lceil \frac{l}{2}\rceil-1$.

In this case $d_{up}(j,k)$ is $\lambda_{2i+1}$ follows directly from the definition in \eqref{rcmatrix} and construction of AIR matrix. Figure \ref{afig112} illustrates this.

\begin{figure}
\centering
\includegraphics[scale=0.80]{matrix121.eps}\\
\caption{up-distance calculation}
\label{afig112}
\end{figure}
\textit{Case (ii)}:$\mathbf{L}(j,k) \in \mathbf{I}_{\lambda_{2i} \times \beta_{2i} \lambda_{2i}}$ (one of the even-submatrices) for $i=0,1,2,\ldots,\lfloor \frac{l}{2}\rfloor$

From the definition of up distance, we have $L(j-d_{up}(j,k),k) \in \mathbf{I}_{\beta_{2i-1}\lambda_{2i-1} \times \lambda_{2i-1}}$.
\begin{figure*}
\centering
\includegraphics[scale=1.20]{VD1.eps}\\
\caption{up-distance calculation}
\label{afig11}
\end{figure*}

We have $k_R=k~\text{mod}~(K-D-\lambda_{2i-1})=c\lambda_{2i}+d$ for some positive integers $c$ and $d$ $(d<\lambda_{2i})$. 
From Figure \ref{afig11}, we have 
\begin{align*}
&d_{up}(k)=d_1+d_2,\\&
d_1=\lambda_{2i-1}-k_R~\text{and}~\\&
d_2=d=k_R-c\lambda_{2i}.
\end{align*}

Hence, we have 
\begin{align*}
d_{up}(k)&=\underbrace{\lambda_{2i-1}-k_R}_{d_1}+\underbrace{k_R-c\lambda_{2i}}_{d_2}=\lambda_{2i-1}-c\lambda_{2i}.
\end{align*}

\section*{APPENDIX D}
\subsection*{Proof of Lemma \ref{lemma3}}
\textit{Case (i)}:  $\mathbf{L}(j,k) \in \mathbf{I}_{\lambda_{2i} \times \beta_{2i} \lambda_{2i} }$  for $i\in [0:\lfloor \frac{l}{2}\rfloor]$ and $k_R \in [0:(\beta_{2i}-1)\lambda_{2i}-1]$. 

In this case $d_{right}(j,k)$ is $\lambda_{2i}$ follows directly from the definition in \eqref{rcmatrix}. 

\textit{Case (ii)}:  $\mathbf{L}(j,k) \in \mathbf{I}_{\lambda_{2i} \times  \beta_{2i} \lambda_{2i}}$  for $i\in [0:\lfloor \frac{l}{2}\rfloor]$ and $k_R \in [(\beta_{2i}-1)\lambda_{2i}:\beta_{2i} \lambda_{2i}-1]$. 

In this case, from the definition of right-distance, we have $L(j, k+d_{right}(j,k)) \in \mathbf{I}_{\beta_{2i+1}\lambda_{2i+1} \times \lambda_{2i+1}}$. We have $k_R=k~\text{mod}~(K-D-\lambda_{2i})=c\lambda_{2i+1}+d$ for some positive integers $c$ and $d$ $(d<\lambda_{2i+1})$. 
From Figure \ref{afig21}, we have 
\begin{align*}
&d_{right}(j,k)=d_1+d_2,\\&
d_1=\lambda_{2i}-d_3, \\&
d_3=k_R~\text{and}~d=d_2.
\end{align*}
Hence, 
\begin{align*}
d_{right}(j,k)=\underbrace{\lambda_{2i}-k_R}_{d_1}+\underbrace{d}_{d_2}&=\lambda_{2i}-c\lambda_{2i+1}-d+d \\&
=\lambda_{2i}-c\lambda_{2i+1}.
\end{align*}
\begin{figure*}
\centering
\includegraphics[scale=1.02]{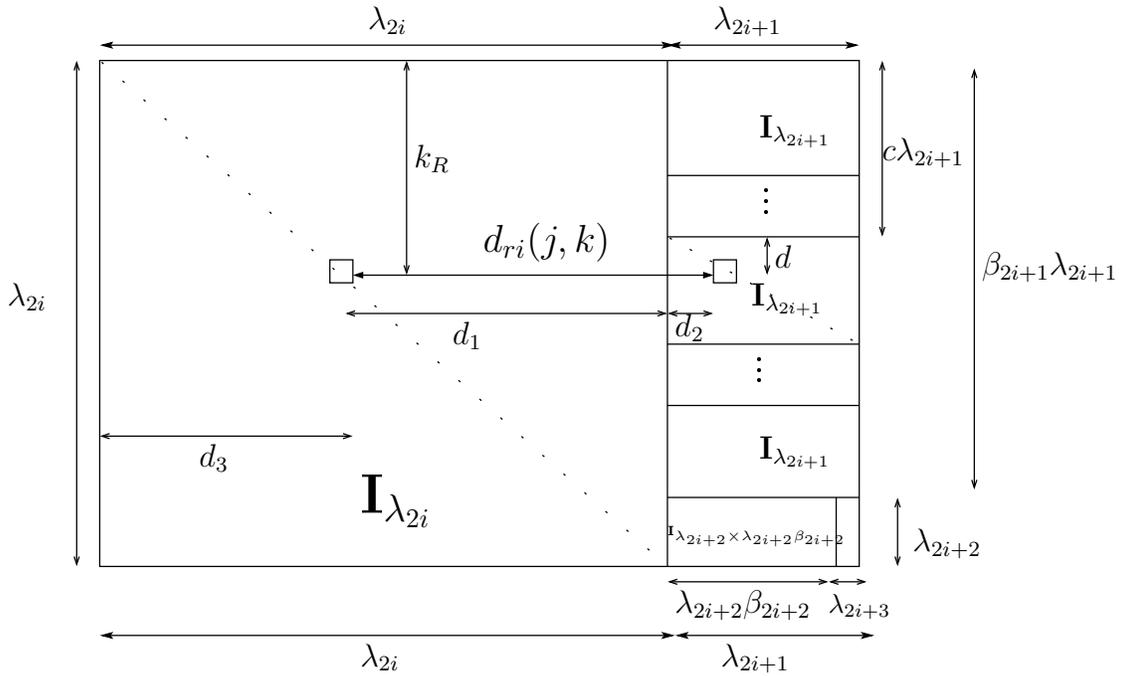}\\
\caption{right distance calculation}
\label{afig21}
\end{figure*}
\section*{APPENDIX E}
\subsection*{Proof of Theorem \ref{thm2}}

{\bf Case (i):} $k \in D_{i}$  for $i=0,1,2,..\lceil \frac{l}{2}\rceil$.

In this case, we have $k_R \in [0:(\beta_{2i}-1)\lambda_{2i}-1]$  for $i=0,1,2,..\lfloor \frac{l}{2}\rfloor$. Let $k_R=c\lambda_{2i}+d$ for some positive integers $c$ and $d$ and $d<\lambda_{2i}$. From Lemma \ref{lemma3}, we have $\mu_{k}=\lambda_{2i}$. From Lemma \ref{lemma1}, we have 
\begin{align}
\label{fact1}
d_{down}(k)&=D+\lambda_{2i+1}+(\beta_{2i}-1-c)\lambda_{2i}\\&
\nonumber 
=D+\lambda_{2i-1}-(c+1)\lambda_{2i}.
\end{align}

From Lemma \ref{lemma2}, we have 
\begin{align}
\label{fact2}
\nonumber
&d_{up}(k+d_{down}(k),k)=\lambda_{2i-1}-c\lambda_{2i}~~\text{and}\\& d_{up}(k+d_{down}(k),k+\mu_{k})=\lambda_{2i-1}-(c+1)\lambda_{2i}.
\end{align}

In this case, we have  $k+\mu_k=k+\lambda_{2i} \in C_i$, hence $\mathbf{L}(k+d_{down}(k),k+\mu_k) \in \mathbf{I}_{\lambda_{2i} \times \beta_{2i}\lambda_{2i}}$ and $p_k=0$. 

From \eqref{fact1} and \eqref{fact2}
\begin{align}
\label{fact3}
d_{down}(k)-d_{up}(k&+d_{down}(k),k+\mu_{k})=D,
\end{align}
this indicates that all message symbols in $c_{k+\mu_{k}}$ are in the side-information of $R_{k}$ except $x_{k+d_{down}(k)}$. From \eqref{fact2} and \eqref{fact3}, every message symbol in $c_{k}+c_{k+\mu_{k}}$ is in the side-information of $R_{k}$. From \eqref{code}, we can write $c_k$ as
\begin{align*}
c_k=x_k+\underbrace{\sum_{z=1}^{i}\sum_{j=1}^{\beta_{2z-1}}x_{k+D-\lambda_{2z-2}+j\lambda_{2z-1}}}_{\nu_k}+x_{k+d_{down}(k)}.
\end{align*}

Hence, the term $\nu_k$ in \eqref{case11} is the exclusive OR of the set of message symbols present in $c_k$ excluding the message symbols $x_k$ and $x_{k+d_{down}(k)}$. The term $\nu_{k+\mu_k}$ in \eqref{case11} is the exclusive OR of the set of message symbols present in $c_{k+\mu_k}$ excluding the message symbol $x_{k+d_{down}(k)}$. From \eqref{fact2} and \eqref{fact3}, $\nu_{k}$ and $\nu_{k+\mu_k}$ are known to receiver $R_k$. Hence, $R_k$ can decode $x_k$ from \eqref{case1}. 

{\bf Case (ii):} $k \in E_{i}$  for $i=0,1,2,..\lceil\frac{l}{2}\rceil-1$.

In this case, we have $k_R \in [(\beta_{2i}-1)\lambda_{2i}:\beta_{2i}\lambda_{2i}-1]$ for $i=0,1,2,..\lfloor \frac{l}{2}\rfloor-1$. Let $k_R=(\beta_{2i}-1)\lambda_{2i}+c\lambda_{2i+1}+d$ for some positive integers $c,d \ (d<\lambda_{2i+1})$. We have $k=K-D-\lambda_{2i-1}+k_R$. From Lemma \ref{lemma1}, we have 
\begin{align}
\label{fact5}
\nonumber
d_{down}(k)&=D+\lambda_{2i+1}+(\beta_{2i}-1-(\beta_{2i}-1))\lambda_{2i}
\\&=D+\lambda_{2i+1}. 
\end{align}
From Lemma \ref{lemma3}, we have 
\begin{align}
\label{fact6}
\nonumber
\mu_{k}&=d_{right}(k+d_{down}(k),k)\\&
\nonumber
=d_{right}(K-D-\lambda_{2i-1}+k_R+D+\lambda_{2i+1},k)\\&
\nonumber
=d_{right}(K-\lambda_{2i}+c\lambda_{2i+1}+d,k)\\&
=\lambda_{2i}-c\lambda_{2i+1}.
\end{align}
From Lemma \ref{lemma2}, we have
\begin{align}
\label{fact7}
d_{up}(k+d_{down}(k),k+\mu_{k})=\lambda_{2i+1}.
\end{align}
From \eqref{fact5} and \eqref{fact7}
\begin{align}
\label{fact8}
d_{down}(k)-d_{up}(k+d_{down}(k),k+\mu_{k})=D,
\end{align}
this indicates that all other message symbols in $c_{k+\mu_{k}}$ are in the side-information of $R_{k}$ except $x_{k+d_{down}(k)}$ and $x_{k+t_{k,\tau}+d_{down}(k)}$ for $\tau=1,2,\ldots,p_{k}$. 
 
From Definition \ref{def1}, $k+d_{down}(k)+t_{k,p_{k}}$ is always less than the number of rows in the matrix $\mathbf{L}$. That is, $k+t_{k,p_k}+d_{down}(k)<K$. Hence, we have 
\begin{align}
\label{fact9}
\nonumber
t_{k,p_k}&< K-k-d_{down}(k)\\&
\nonumber
=K-(K-D-\lambda_{2i-1}+(\beta_{2i}-1)\lambda_{2i}+c\lambda_{2i+1}+d)-\\&
\nonumber
~~~~~~~~~~~~~~~~~~~~~~~~~~~~~~~~(D+\lambda_{2i+1})\\&
=\lambda_{2i}-c\lambda_{2i+1}-d.
\end{align}
From \eqref{fact6} and \eqref{fact9}
\begin{align}
\label{relation2}
t_{k,p_{k}} < \mu_{k}-d.
\end{align}
From \eqref{fact9}, we have 
\begin{align*}
k_R+t_{k,p_{k}}&<k_R+\lambda_{2i}-c\lambda_{2i+1}-d\\&=\underbrace{(\beta_{2i}-1)\lambda_{2i}+c\lambda_{2i+1}+d}_{k_R}+\lambda_{2i}-c\lambda_{2i+1}-d\\&=\beta_{2i}\lambda_{2i}.
\end{align*}
Hence,
\begin{align*}
k_R+t_{k,p_{k}} \in [(\beta_{2i}-1)\lambda_{2i}:\beta_{2i}\lambda_{2i}-1]
\end{align*}
and
\begin{align}
\label{relation3}
\mathbf{L}(k+t_{k,\tau}+d_{down}(k+t_{k,\tau}),k+t_{k,\tau}) \in \mathbf{C}_{\lambda_{2i} \times  \beta_{2i} \lambda_{2i}}
\end{align}
for $\tau=1,2,\ldots,p_{k}$. Hence, we have  
\begin{align}
d_{down}(k)=d_{down}(k+t_{k,\tau})
\end{align}
for $\tau=1,2,\ldots,p_{k}$. 

From Lemma \ref{lemma2} and \eqref{relation3}, for $\mathbf{L}(k+t_{k,\tau}+d_{down}(k+t_{k,\tau}),k+t_{k,\tau})$, we have,
\begin{align}
\label{fact10}
d_{up}(k+t_{k,\tau}+d_{down}(k+t_{k,\tau}),k+t_{k,\tau})=\lambda_{2i}+\lambda_{2i+1}.
\end{align}
From \eqref{fact7} and \eqref{fact10}, we have 
\begin{align}
\label{fact15}
\nonumber
&d_{up}(k+t_{k,\tau}+d_{down}(k+t_{k,\tau}),k+t_{k,\tau})-\\&d_{up}(k+d_{down}(k),k+\mu_{k})=\lambda_{2i},
\end{align}
which along with \eqref{fact8} indicates that every message symbol in $c_{k+t_{k,\tau}}$ is in the side-information of $R_{k}$ except $x_{k+t_{k,\tau}+d_{down}(k+t_{k,\tau})}$ for $\tau=1,2,\ldots,p_k$. 


We have $k_R \in [(\beta_{2i}-1)\lambda_{2i}:\beta_{2i}\lambda_{2i}-1]$, $d_{up}(k+d_{down}(k),k)=\lambda_{2i}+\lambda_{2i+1}$. From \eqref{fact10}, we have
\begin{align}
\label{fact11}
\nonumber
d_{up}(k+d_{down}(k),k)-d_{up}(k&+d_{down}(k),k+\mu_{k})\\&=\lambda_{2i},
\end{align}
and this along with \eqref{fact8} indicates that every message symbol in $c_{k}$ is in the side-information of $R_{k}$ except $x_{k+d_{down}(k)}$. 

From \eqref{fact8},\eqref{fact15} and \eqref{fact11}, the interfering message symbol $x_{k+t_{k,\tau}+d_{down}(k)}$ in $c_{k+\mu_{k}}$ can be canceled by adding the broadcast symbol $c_{k+t_{k,\tau}}$ for $\tau=1,2,\ldots,p_k$  and the interfering message symbol $x_{k+d_{down}(k)}$ in $c_{k+\mu_{k}}$ can be canceled by adding the broadcast symbol $c_{k}$. Hence, $R_k$ can decode the message symbol $x_k$ by adding the broadcast symbols $c_{k},c_{k+\mu_{k}}$ and  $c_{k+t_{k,\tau}}$ for $\tau=1,2,\ldots,p_{k}$. 

The term $\nu_k$ in \eqref{case21} is the exclusive OR of the set of message symbols present in $c_k$ excluding the message symbols $x_k$ and $x_{k+d_{down}(k)}$. The term $\nu_{k+t_{k,\tau}}$ for $\tau \in [1:p_k]$ in \eqref{case21} is the exclusive OR of the set of message symbols present in $c_{k+t_{k,\tau}}$ excluding the message symbol $x_{k+t_{k,\tau}+d_{down}(k+t_{k,\tau})}$. The term $\nu_{k+\mu_k}$ in \eqref{case21} is the exclusive OR of the set of message symbols present in $c_{k+\mu_k}$ excluding the message symbols $x_{k+d_{down}(k)}$ and $x_{k+t_{k,\tau}+d_{down}(k+t_{k,\tau})}$ for $\tau \in [1:p_k]$. From \eqref{fact8}, \eqref{fact15} and \eqref{fact11}, $\nu_{k},\nu_{k+\mu_k}$ and $\nu_{t_{k,\tau}}$ for $\tau \in [1:p_k]$ are known to receiver $R_k$. Hence, $R_k$ can decode $x_k$ from \eqref{case2}. 
{\bf Case (iii):} $k \in [K-D-\lambda_l:K-D-1]= E_{i}$  for $i=\left\lceil\frac{l}{2}\right\rceil$.

In this case, from Lemma \ref{lemma1}, we have 
\begin{align}
\label{fact51}
d_{down}(k)=D.
\end{align}
Hence, every message symbol in $c_k$ is in the side-information of $R_k$ excluding the message symbol $x_k$. The term $\nu_k$ in \eqref{case3} is the exclusive OR of the set of message symbols present in $c_k$ excluding the message symbol $x_{k}$. From \eqref{fact51}, $\nu_{k}$ is known to receiver $R_k$. Hence, $R_k$ can decode $x_k$ from \eqref{case3}.

{\bf Case (iv):} $k \in [K-D:K-1]$

\begin{itemize}
\item  If $D < \left\lceil\frac{K}{2}\right \rceil$, from \eqref{code}, the broadcast symbol $c_{k~\text{mod}~(K-D)}$ is given by $c_{k~\text{mod}~(K-D)}=x_{k~\text{mod}~(K-D)}+x_{k~\text{mod}~(K-D)+K-D}=x_{k+D}+x_{k}$. In $c_{k~\text{mod}~(K-D)}$, the message symbol $x_{k+D}$ is in the side-information of receiver $R_{k}$. Hence, $R_{k}$ can decode its wanted message symbol $x_{k}$ from $c_{k~\text{mod}~(K-D)}$. 
\item If $D \geq \left \lceil \frac{K}{2}\right \rceil$, from \eqref{chain}, we have $\beta_0=0$ and $\lambda_1=K-D$. From Lemma \ref{lemma2}, we have 
\begin{align}
\label{fact52}
d_{up}(k,k~\text{mod}~(K-D))=\lambda_1=K-D.
\end{align}

Hence, $c_{k~\text{mod}~(K-D)}$ does not contain message symbols from the set $$\{x_{k-(K-D)+1},x_{k-(K-D)+2},\ldots,x_{k-1}\}=\mathcal{I}_k$$(interference to receiver $R_k$). Every message symbol in $c_{k^\prime}$ is in the side-information of $R_k$ excluding the message symbol $x_k$. The term $\nu_{k^{\prime}}$ in \eqref{case4} is the exclusive OR of the set of message symbols present in $c_{k^{\prime}}$ excluding the message symbol $x_{k}$. From \eqref{fact52}, $\nu_{k^{\prime}}$ is known to receiver $R_k$. Hence, $R_k$ can decode $x_k$ from \eqref{case4}.

\end{itemize}
This completes the proof.
\section*{Acknowledgment}
This work was supported partly by the Science and Engineering Research Board (SERB) of Department of Science and Technology (DST), Government of India, through J.C. Bose National Fellowship to B. Sundar Rajan

\end{document}